\documentclass[a4paper]{article}   

\usepackage[utf8]{inputenc}
\usepackage[T1]{fontenc}

\usepackage{amsthm,amssymb,amsmath,enumerate}
\usepackage{algorithm}

\usepackage{tikz}
\usetikzlibrary{backgrounds}
\usetikzlibrary{decorations}
\usetikzlibrary{decorations.pathmorphing}

\usepackage{xspace}

\usepackage{graphicx}

\usepackage{tikz}

\usepackage{graphics}
\usepackage{longtable}
\usepackage{array} 

\usepackage{adjustbox}

\usetikzlibrary{calc}

\usepackage[american]{babel}

\usepackage{mathptmx}      
\newcommand{\vts}[1]{V(#1)}
\newcommand{\eds}[1]{E(#1)}

\newcommand{\vn}{\varnothing}
\newcommand{\se}{\subseteq}
\newcommand{\sm}{\setminus}
\newcommand{\cro}[1]{\left[#1\right]}
\newcommand{\ceil}[1]{\lceil #1\rceil}
\newcommand{\floor}[1]{\lfloor #1\rfloor}
\newcommand{\rset}{\mathbb{R}}

\newcommand{\stab}[1]{\mathsf{STAB}(#1)}

\newcommand{\comp}[1]{\overline{#1}}

\newcommand{\mntc}{4-critical\xspace}

\newcommand{\set}[1]{\left\{#1\right\}}

\newcommand{\cl}[1]{\overline{L(#1)}}

\newcommand{\stc}{{\sf ST}-cover\xspace}

\newtheorem{definition}{Definition}
\newtheorem{proposition}[definition]{Proposition}
\newtheorem{theorem}[definition]{Theorem}
\newtheorem{corollary}[definition]{Corollary}
\newtheorem{lemma}[definition]{Lemma}

\usepackage{tikz}
\usepackage{etoolbox}
\AtBeginEnvironment{theorem}{\setcounter{equation}{0}}

\begin{document}

\title{On 4-critical $t$-perfect graphs}

\author{Yohann Benchetrit\thanks{Universit\'e Libre de Bruxelles, Department of Mathematics, Algebra and Combinatorics, partially supported by ARC Grant AUWB-2012-12/17-ULB2-COPIHYMA ({\tt yohann.benchetrit@ulb.ac.be})}}

\date{}

\maketitle

\begin{abstract}
	It is an open question whether the chromatic number of $t$-perfect graphs is bounded by a constant. The largest known value for this parameter is 4, and the only example of a 4-critical $t$-perfect graph, due to Laurent and Seymour, is the complement of the line graph of the prism $\Pi$ (a graph is 4-critical if it has chromatic number 4 and all its proper induced subgraphs are 3-colorable).
	
	In this paper, we show a new example of a 4-critical $t$-perfect graph: the complement of the line graph of the 5-wheel $W_5$.
	Furthermore, we prove that these two examples are in fact the only 4-critical $t$-perfect graphs in the class of complements of line graphs.
	As a byproduct, an analogous and more general result is obtained for $h$-perfect graphs in this class. 
	
	The class of $P_6$-free graphs is a proper superclass of complements of line graphs and appears as a natural candidate to further investigate  the chromatic number of $t$-perfect graphs. We observe that a result of Randerath, Schiermeyer and Tewes implies that every $t$-perfect  $P_6$-free graph is 4-colorable. Finally, we use results of Chudnovsky et al to show that $\cl{W_5}$ and $\cl{\Pi}$ are also the only 4-critical $t$-perfect $P_6$-free graphs. 

\end{abstract}

\section{Introduction}\label{intro}

The original definition  of a perfect graph (due to Berge) is, at least in appearance, a purely combinatorial condition on the gap between the chromatic number and one of its natural lower bounds; a graph is \emph{perfect} if for each of its induced subgraphs, the chromatic number and the clique number are equal.

Fundamental works of Lov\'asz \cite{Lovasz1972} and Fulkerson \cite{Fulkerson1972} imply, as stated by Chv\'atal \cite{Chvatal1975}, that perfect graphs are characterized by the structure of their stable set polytope (i.e the convex hull of their stable sets): a graph is perfect if and only if the facets of its stable set polytope are defined by non-negativity and clique inequalities (see Section \ref{defs} for precise definitions).

This polyhedral condition may be relaxed in various ways, each leading to a different generalization of the notion of perfection (see for example \cite{Shepherd1994,Wagler2002}). 
A graph is \emph{$h$-perfect} if every facet of its stable set polytope is  defined by a non-negativity or  clique inequality, or by an odd circuit inequality. A graph is \emph{$t$-perfect} if it is $h$-perfect and $K_4$-free. 
The polyhedral characterization of perfection easily shows that  perfect graphs are $h$-perfect. 

Obviously, each $K_4$-free perfect graph is 3-colorable. Yet, it is still open whether the chromatic number of $t$-perfect graphs is bounded by a constant.

We say that a graph is \emph{4-critical} if it is not 3-colorable, and each of its proper induced subgraphs is 3-colorable. Laurent and Seymour showed a 4-critical $t$-perfect graph: the complement of the line graph of the prism $\Pi$ (Figure \ref{PrismAndLine}). Its chromatic number is 4 (as for every 4-critical graph), which is currently the largest value known for the chromatic number of a $t$-perfect graph. 

\begin{figure}
	\centering
	\begin{tikzpicture}[scale=1.5]
	\tikzset{vertex/.style={draw,circle,scale=0.3}}
	\tikzset{texts/.style={scale=0.7}}
	
	\begin{scope}[xshift=-3cm,yshift=-.2cm]
	\node[vertex] (A) at (-.5,1) {};
	\node[vertex] (B) at (.5,1) {};
	\node[vertex] (C) at (0,0.75) {};
	
	\node[vertex] (D) at (-.5,-.5) {};
	\node[vertex] (E) at (.5,-.5) {};
	\node[vertex] (F) at (0,-.25) {};
	
	\draw (A)--(B)--(C)--(A);
	\draw (D)--(E)--(F)--(D);
	\draw (A)--(D);
	\draw (C)--(F);
	\draw (B)--(E);
	
	\end{scope}

	\begin{scope}

	\node[vertex] (C) at (90:0.25) {};
	\node[vertex] (E) at (210:0.25) {};
	\node[vertex] (F) at (330:0.25) {};
	\node[vertex] (G) at (-15:.7) {};
	\node[vertex] (I) at (-45:.7) {};
	\node[vertex] (D) at (195:.7) {};
	\node[vertex] (H) at (225:.7) {};
	\node[vertex] (A) at (105:.7) {};
	\node[vertex] (B) at (75:.7) {};
	
	\draw (C)--(E)--(F)--(C);
	\draw (A)--(B)--(G)--(I)--(H)--(D)--(A);
	\draw (A)--(C)--(B);
	\draw (G)--(F)--(I);
	\draw (D)--(E)--(H);
	
	\draw (A)to[bend left=90,min distance=1cm] (I);
	\draw (B)to[bend right=90,min distance=1cm](H);
	\draw (D)to[bend right=90,min distance=1cm](G);
	
	\end{scope}
	
	%

	\end{tikzpicture}
	\caption{$\Pi$ and $\cl{\Pi}$}\label{PrismAndLine}
\end{figure}

It is easy to check that $t$-perfect graphs whose complement is a line graph have chromatic number at most 4 (see Section \ref{proofs}). These facts motivate searching for possibly more 4-critical $t$-perfect graphs in the class of complements of line graphs.

The first main result of this paper is a new example of a \mntc $t$-perfect  graph: the complement of the line graph of the wheel $W_5$ (Figure \ref{clw5}).

\begin{figure}[h]
	\centering
	\begin{tikzpicture}[scale=1.5]
		\tikzset{vertex/.style={draw,circle,scale=0.3}}
		
		\tikzset{vertex/.style={draw,circle,scale=0.3}}
		\tikzset{texts/.style={scale=0.7}}
		
		\begin{scope}
		\node[vertex] (A) at (0,1) {};
		\node[vertex] (B) at (162:1) {};
		\node[vertex] (C) at (234:1) {};
		\node[vertex] (D) at (306:1) {};
		\node[vertex] (E) at (378:1) {};
		
		\node[vertex] (Cent) at (0,0) {}; 
		
		\draw (A) -- (B) -- (C) -- (D) -- (E) -- (A);
		
		\draw (Cent) -- (A);
		\draw (Cent) -- (B);
		\draw (Cent) -- (C);
		\draw (Cent) -- (D);
		\draw (Cent) -- (E);
		
		
		\end{scope}
		
		\begin{scope}[xshift=2.5cm]
		\node[vertex] (H) at (0,0) {};
		\node[vertex] (E) at (45:.5) {};
		\node[vertex] (I) at (-45:.5) {};
		\node[vertex] (D) at (135:.5) {};
		\node[vertex] (G) at (225:.5) {};
		
		\node[vertex] (B) at (55:1) {};
		\node[vertex] (F) at (15:1){};
		\node[vertex] (A) at (125:1){};
		\node[vertex] (C) at (165:1){};
		\node[vertex] (J) at (0,-.8){};
		
		\draw (C)--(D)--(B)--(F);
		\draw (F)--(E)--(A)--(C);
		\draw (C)to[bend left=90, min distance=1cm] (F);
		\draw (A)--(D);
		\draw (B)--(E);
		
		\draw (E)--(H)--(I)--(E);
		\draw (D)--(H)--(G)--(D);
		\draw (H)--(J);
		
		\draw (C)to[bend right=90,min distance=.5cm](J);
		\draw (J)to[bend right=90,min distance=.5cm](F);
		
		\draw (C)to[out=-90,in=180]($(G)+(0,-.35)$)to[out=0,in=225](I);
		\draw (G)to[out=-45,in=180] ($(I)+(0,-.35)$) to[out=0,in=-90](F);
		
		
		\end{scope}
		

%
%
%
%
%
%
%
%
%
%
	\end{tikzpicture}
	\caption{$W_5$ and $\cl{W_5}$}\label{clw5}
\end{figure}

To the best of our knowledge, $\cl{\Pi}$ and $\cl{W_5}$ are the only known available examples of \mntc $t$-perfect graphs. 
The other main result of the paper states that, in fact, no more such graphs can be found in the class of graphs whose complement is a line graph:

\begin{theorem}\label{colcltp}
	$\cl{\Pi}$ and $\cl{W_5}$ are the only 4-critical $t$-perfect graphs in the class of complements of line graphs.
\end{theorem}

Using an argument of Seb\H{o} (see \cite[Lemma 26]{Bruhn2012}), we easily obtain the following (strictly) stronger result as a corollary. Let $\chi_f(G)$ denote the fractional chromatic number of a graph $G$ (see Section \ref{defs}).

\begin{corollary}\label{colclhp}
	For every $h$-perfect graph $G$ which is the complement of a line graph, the following statements are equivalent:
	\begin{itemize}
		\item [i)] every induced subgraph $H$ of $G$ satisfies $\chi(H)=\ceil{\chi_f(H)}$;
		\item [ii)] $G$ does not contain $\cl{\Pi}$ or $\cl{W_5}$ as an induced subgraph.
	\end{itemize}
	\end{corollary}
A corresponding coloring or an obstruction can be easily found in polynomial-time. Since the complexity of testing $h$-perfection (and even $t$-perfection) is still open, this raises the problem of how an $h$-perfect graph is given as input to the algorithm. In fact, there is no ambiguity since we will show that \emph{$h$-perfection can be (easily) tested in polynomial-time for graphs which are the complement of a line graph}.

Every $t$-perfect graph $G$ satisfies $\chi_f(G)\leq 3$ (see Section \ref{defs}), and therefore Corollary \ref{colclhp} yields Theorem \ref{colcltp}, when specialized to $K_4$-free graphs. 
\vspace{5pt}

A graph is \emph{$P_k$-free} if it does not contain the path $P_k$ with $k$ vertices as an induced subgraph. It is straightforward to check that the complement of $P_6$ is not a line graph. Hence, complements of line graphs form a (proper) subclass of $P_6$-free graphs (whose structure is considerably richer and more difficult to study; see  \cite{ChudnovskyGoedgebeurSchaudtEtAl2015}). 
Since the only known examples of 4-critical $t$-perfect graphs are complements of line graphs, it is tempting to further investigate the chromatic number of $t$-perfect graphs in the class of $P_6$-free graphs.

We first observe that a result of Randerath, Schiermeyer and Tewes \cite{Randerath2002} implies that  $P_6$-free graphs  are of no use in finding a $t$-perfect graph of chromatic number greater than 4 (whether such a graph exists is open). More generally, we show:

\begin{theorem}\label{p6free}
	Every $h$-perfect $P_6$-free graph $G$ is $(\omega(G)+1)$-colorable.
\end{theorem}

Since $\cl{\Pi}$ and $\cl{W_5}$ are $P_6$-free, this bound is tight.
In \cite{ChudnovskyGoedgebeurSchaudtEtAl2015}, Chudnovsky et al showed the explicit list of all \mntc $P_6$-free graphs (there, \mntc graphs are called \emph{4-vertex-critical}). There are exactly 80 such graphs, and the largest one has 16 vertices. 
It is not difficult to design a computer program  that tests whether a graph is $t$-perfect (with the help of \cite{sage}, for example), and whose running time remains practical for graphs with no more than 16 vertices (recall that the computational complexity of testing $t$-perfection is open). 
Using such a program, we could check: 

\begin{theorem}\label{p6critical}
	$\cl{\Pi}$ and $\cl{W_5}$ are the only \mntc $t$-perfect $P_6$-free graphs.
\end{theorem} 

We prove Theorem \ref{p6critical} as follows: the $t$-perfection of $\cl{\Pi}$ and $\cl{W_5}$ is easy to check (see Section \ref{proofs}) and, for each of the other seventy-eight \mntc $P_6$-free graphs, we provide a non-integral vertex of the polytope defined by non-negativity, clique and odd circuit inequalities (as a concise certificate of  non-$t$-perfection). 

Since complements of line graphs are $P_6$-free, Theorem \ref{colcltp} appears as a direct corollary of \cite{ChudnovskyGoedgebeurSchaudtEtAl2015}. Yet, this approach is considerably more complicated than needed, as it requires knowing explicitly the set of \mntc $P_6$-free graphs (more precisely, it needs knowing a characterization of the facets of their stable set polytope); determining this set in such a way involves some rather deep structural analysis as well as a significant amount of computer-based enumeration \cite{ChudnovskyGoedgebeurSchaudtEtAl2015}. The proof of Theorem \ref{colcltp} that we present in Section \ref{proofs} is significantly simpler and hopefully more specific of $t$-perfection (the interest of such a specific approach is discussed further in Section \ref{conclusion}).


\paragraph{Related results} Chv\'atal introduced $t$-perfect graphs in \cite{Chvatal1975}; see \cite[chap. 68]{Schrijver2003} for a survey of the topic. Bruhn and Stein \cite{Bruhn2012} showed that $t$-perfect claw-free graphs are 3-colorable, and derived as a corollary that  every $h$-perfect claw-free graph $G$ satisfies $\chi(G)=\ceil{\chi_f(G)}$ (a graph is \emph{claw-free} if no vertex has three pairwise-non-adjacent neighbors; line graphs form a subclass of claw-free graphs). 
\cite{Bruhn2014} shows a polynomial-time algorithm testing $t$-perfection in claw-free graphs. Also, deciding $h$-perfection in line graphs can be carried out in polynomial-time \cite{Benchetrita}.

Bruhn and Fuchs \cite{BruhnFuchs2015} proved that every $t$-perfect $P_5$-free graph is 3-colorable, by providing a certificate of non-$t$-perfection for each of the twelve \mntc $P_5$-free graphs (which are described explicitly by Maffray and Morel in \cite{maffray20123}). 


Computing the chromatic number of the complement of a line graph is {\sc NP}-hard, since it contains the Vertex-cover Problem for triangle-free graphs as a special case \cite{Poljak1974}. By contrast, deciding 3-colorability in  the class of $P_6$-free graphs (which contains the class of complements of line graph) can be carried out in polynomial-time \cite{Randerath2004}.


\paragraph{Outline} Definitions and basic useful facts are given in Section 2. In Section 3, we show that $\cl{W_5}$ is a \mntc $t$-perfect graph and prove Theorem \ref{colcltp} and Corollary \ref{colclhp}. We also show that $h$-perfection can be tested efficiently for graphs whose complement is a line graph, and derive Theorem \ref{p6free}. Further consequences of our results and open problems are deferred to Section 4. The proof of Theorem \ref{p6critical} is deferred to Appendix \ref{appen}, as it is  of slightly different nature than the other proofs of the paper.

\paragraph{Acknowledgments} We are thankful to Andr\'as Seb\H{o} for his precious remarks and comments which helped improving the presentation of the manuscript.

\section{Definitions and useful facts}\label{defs}
We work with finite undirected graphs only; they may have multiple edges but no loops. We follow \cite{Schrijver2003} for standard terms and notations that are not defined in the paper.

Let $G$ be a graph. We say that $G$ is \emph{simple} if it has no multiple edge (that is an edge which has at least one other parallel edge).  For each $X\se\vts{G}$, let $G\cro{X}$ denote the subgraph of $G$ defined by $V(G\cro{X}):=X$ and $E(G\cro{X})$ is the set of edges of $G$ whose ends both belong to $X$; $G\cro{X}$ is called the subgraph of $G$ \emph{induced by $X$}, and any subgraph of this form is an \emph{induced subgraph} of $G$. A subgraph $H$ of $G$ is \emph{spanning} if $V(H)=V(G)$.
We say that $G$ \emph{contains} a graph $H$ (resp. as an induced subgraph) if $G$ has a subgraph (resp. induced subgraph) isomorphic to $H$, and that $G$ is \emph{$H$-free} if it does not contain $H$ as an induced subgraph.

For each vertex $v$ of $G$, let $\delta_G(v)$ denote the set of edges of $G$ incident to $v$; the set $\delta_G(v)$ is called the \emph{full star at $v$}. A \emph{star} is a subset of a full star. The \emph{degree} of a vertex $v$ in $G$ is  $|\delta_G(v)|$.

A \emph{circuit} is a connected graph whose vertices all have degree 2, and a \emph{path} is a graph obtained by deleting a single edge from a circuit; notice that our paths have two distinct ends. The \emph{length} of a circuit (or a path) is its number of edges. A circuit or path is \emph{odd} if it has odd length.
The \emph{complement} of $G$, denoted by $\comp{G}$, is the simple graph with the same vertex-set as $G$ and whose edges are the non-edges of $G$. The \emph{line graph} of $G$, denoted by $L(G)$, is the graph whose vertex-set is $E(G)$ and such that two edges of $G$ are adjacent in $L(G)$ if and only if they have at least one common end.

A \emph{stable} set of  $G$ is a set of pairwise non-adjacent vertices. The largest cardinality of a stable set is written $\alpha(G)$. A \emph{clique} of $G$ is a stable set of $\comp{G}$, and we write $\omega(G):=\alpha(\comp{G})$.
A \emph{$k$-coloring} of $G$ is a set of $k$ stable sets of $G$ whose union contains every vertex of $G$, and  $G$ is \emph{$k$-colorable} if it has at least one $k$-coloring. The chromatic number of $G$, denoted by $\chi(G)$, is the smallest $k$ for which $G$ admits a $k$-coloring. Obviously, $\omega(G)\leq \chi(G)$. We say that $G$ is \emph{perfect} if each induced subgraph $H$ of $G$ satisfies $\chi(H)=\omega(H)$. 
A graph is \emph{$k$-critical} if it is not $(k-1)$-colorable and all its proper induced subgraphs are $(k-1)$-colorable. Clearly, the chromatic number of such a graph is $k$.

The incidence vector of a set $S\se\vts{G}$, denoted by $\chi^{S}$, is the vector of $\rset^{\vts{G}}$ defined for every $v\in\vts{G}$ by: $\chi^{S}(v)=1$ if $ v\in S $ and $ \chi^{S}(v)=0 $ otherwise. The \emph{stable set polytope} of $G$, denoted by $\stab{G}$, is the convex hull of the incidence vectors of its stable sets. Obviously, each point of $\stab{G}$ satisfies \emph{the non-negativity inequality} $x_v\geq 0$ for every $v\in\vts{G}$, and the \emph{clique inequality} $\sum_{v\in K}x_v\leq 1$ for each clique $K$ of $G$.

Clearly, for each odd circuit $C$ of $G$, the \emph{odd-circuit inequality} $\sum_{v\in V(C)}x_v\leq \frac{|C|-1}{2}$ is satisfied by every point of $\stab{G}$. Let:

\begin{equation*}\label{fonda-eqn-HSTAB}
\mathsf{HSTAB}(G):=\left\{x\in\rset^{\vts{G}}\colon\begin{array}{cc}
x\geq 0,  &  \\ 
\displaystyle \sum_{v\in K}x_v\leq 1  & \text{$\forall K$ clique of $G$,}  \\ 
\displaystyle \sum_{v\in \vts{C}}x_v\leq \dfrac{|\vts{C}|-1}{2}  & \text{$\forall C$ odd circuit of $G$.} \\ 
\end{array}\right\}.
\end{equation*}
A graph $G$ is \emph{$ h $-perfect} if every facet of $\stab{G}$ is defined by a non-negativity, clique, or odd circuit inequality (the \emph{facets} of $\stab{G}$ are in 1-1 correspondence with those inequalities which appear in every description of $\stab{G}$, up to a positive scalar factor; see \cite{Schrijver2003} for further details). In other words, \emph{$G$ is $h$-perfect if and only if $\stab{G}=\mathsf{HSTAB}(G)$}. 

A graph is \emph{$t$-perfect} if it is $h$-perfect and $K_4$-free. Chv\'atal \cite{Chvatal1975} showed that \emph{a graph is perfect if and only if every facet of its stable set polytope is defined by a non-negativity or clique inequality}. This directly implies that \emph{every perfect graph is $h$-perfect}. 

For $n\geq 3$, let $C_n$ denote the circuit of length $n$, and let $W_n$ denote the graph obtained from $C_n$ by adding a new vertex adjacent to every vertex of $C_n$. The graphs $W_{2k+1}$ ($k\geq 1$) are the \emph{odd wheels}. In particular, $W_3=K_4$ is $h$-perfect (since it is perfect), but not $t$-perfect. The \emph{odd holes} are the graphs $C_n$ for $n$ odd and at least 5, and \emph{odd antiholes} are their complements. Notice that $\comp{C_5}$ and $C_5$ are isomorphic. The \emph{prism}, denoted by $\Pi$, is the graph showed in Figure \ref{PrismAndLine} (it is isomorphic to $\comp{C_6}$).

Basic polyhedral arguments show that \emph{$h$-perfection (and $t$-perfection) is closed by taking induced subgraphs}. 
It is easy to check that \emph{odd wheels with at least 5 vertices and odd antiholes with at least 7 vertices are (minimally) $h$-imperfect} \cite[pg. 1194]{Schrijver2003} (further examples of minimally $h$-imperfect graphs are known, though no characterization has been found yet; we do not use them in this paper). Therefore, an $h$-perfect graph cannot contain these as induced subgraphs.

The \emph{fractional chromatic number} of a graph $G$ is the minimum of $\sum_{i=1}^{k}\lambda_i$ for non-negative real numbers $\lambda_1,\ldots,\lambda_k$ for which there exist stable sets $S_1,\ldots,S_k$ of $G$ such that all the coordinates of $\sum_{i=1}^{k}\lambda_i\chi^{S_i}$ are at least 1. It is well-known that $\omega(G)\leq \chi_f(G)\leq \chi(G)$ holds for every graph $G$; each of these inequalities may be strict (as shows the circuit of length 5). Furthermore, \emph{for every graph $G$, $\chi_f(G)=2$ if and only if $G$ is bipartite}.

The duality theorem of linear programming implies the following (see \cite[pg 477]{Bruhn2012} for a proof):

\begin{proposition}\label{fcn}
	For every $h$-perfect graph $G$:
	\[ \chi_f(G)=\max\left(\omega(G),\max\set{\frac{2|V(C)|}{|V(C)|-1}\colon\,\text{$C$ odd circuit of $G$}}\right). \]
\end{proposition}
We use the convention $\max \vn=-\infty$ to keep $\chi_f(G)=2$ when $G$ is bipartite. Clearly, Proposition \ref{fcn} implies that \emph{every $t$-perfect graph $G$ satisfies $\chi_f(G)\leq 3$}.
%
%
%
%
%
%
%
%
%
%
%
%
%
%
%
%
%
%
%

A \emph{matching} of a graph $G$ is a set of pairwise non-incident edges of $G$, and the largest cardinality of a matching is denoted by $\nu(G)$. Clearly, $\nu(G)=\omega(\cl{G})$. A \emph{triangle} of $G$ is the set of edges of a circuit of length 3. An \emph{\stc} of $G$ is a set of stars and triangles whose union contains every edge of $G$. The minimum cardinality of an \stc of $G$ is denoted by $\gamma(G)$. Obviously, every simple graph $G$ satisfies $\chi(\cl{G})=\gamma(G)$.
A triangle with a parallel edge shows that this does not hold for non-simple graphs.

\section{Proofs}\label{proofs}
We first derive a characterization of the $h$-perfection of $\cl{H}$ in terms of $H$ from a result of Shepherd \cite{Shepherd1995}. A \emph{set-join} of a graph $G$ is a set $\set{X_1,\ldots,X_l}$ of pairwise-disjoint (possibly empty) subsets of vertices of $G$ such that for all $1\leq i<j\leq l$ and for each $(u,v)\in X_i\times X_j$, we have $uv\in\eds{G}$. The inequality associated to this set-join is:

\[ \displaystyle\sum_{i=1}^{k}\frac{1}{\alpha(G\cro{X_i})}\displaystyle\sum_{v\in X_i}x_v\leq 1. \]
Obviously, it is valid for $\stab{G}$. Shepherd gave a complete description for the stable set polytope of complements of line graphs: 

\begin{theorem}[\cite{Shepherd1995}]\label{shepherd}
	If $G$ is the complement of a line graph, then each facet of $\stab{G}$ is defined by an inequality associated to a set-join of a clique and some induced odd antiholes of $G$.
\end{theorem}
	
Clearly, any set-join of a clique and $\comp{C_5}$ ($=C_5$) contains $W_5$ as an induced subgraph. Since $W_5$ and odd antiholes with at least 7 vertices are not $h$-perfect, and as $ h $-perfection is closed for taking induced subgraphs (see Section \ref{defs}), Theorem \ref{shepherd} implies: for every graph $G$ which is the complement of a line graph, $G$ is $h$-perfect if and only if it does not contain $W_5$ or an odd antihole with at least 7 vertices as induced subgraphs. This means the following:

\begin{corollary}\label{small}
	For every graph $H$, the following statements are equivalent:
	\begin{itemize}
		\item [i)] $\cl{H}$ is $h$-perfect;
		\item [ii)] every odd circuit of $H$ has length at most 5, and every edge has at least one end in each circuit of length 5.
	\end{itemize}
\end{corollary}

Since $W_5$ and $\Pi$ obviously satisfy ii) and have no matching of cardinality 4, this corollary shows that $\cl{W_5}$ and $\cl{\Pi}$ are $t$-perfect. Observe that every $t$-perfect graph $G$ which is the complement of a line graph has $\chi(G)\leq 4$. Indeed, let $v\in\vts{G}$. Since $G$ is $t$-perfect, it does not contain an odd wheel as an induced subgraph. Thus, the subgraph of $G$ induced by the neighbors of $v$ is bipartite. Besides, as $\comp{G}$ is a line graph, the subgraph of $G$ induced by the non-neighbors of $v$ is also bipartite. Combining respective 2-colorings of these two bipartite graphs directly yields a 4-coloring of $G$.

Seymour and Laurent observed that $\chi(\cl{\Pi})=4$  \cite[pg 1207]{Schrijver2003} (this follows from the same type of argument as below). 
We now prove that $\cl{W_5}$ is another $t$-perfect graph of chromatic number 4:

\begin{proposition}\label{w5}
	$\chi(\cl{W_5})=4$. 
\end{proposition}
	\begin{proof}
	Since $W_5$ is simple, we need only to show that $\gamma(W_5)=4$ (see Section \ref{defs}).
	Let $T$ be any circuit of length 3 of $W_5$. Clearly, $\set{\eds{T}}\cup\set{\delta_{W_5}(v)\colon v\notin \vts{T}}$ is an \stc of $W_5$, hence $\gamma(W_5)\leq 4$. 
	
	It is straightforward to check that $W_5$ cannot be covered by 3 stars. Thus, each \stc of $W_5$ has cardinality at least 4 or  contains a triangle. Besides, for each triangle $T$ of $W_5$, the graph $W_5-\eds{T}$ has a matching of cardinality 3 and hence it cannot have an \stc of cardinality less than 3. This shows $\gamma(W_5)\geq 4$ and the proposition.
	\end{proof}

The following two results are the main ingredients of our proof of Theorem \ref{colcltp}. A graph $G$ is \emph{factor-critical} if for every $v\in\vts{G}$, $G-v$ has a perfect matching.  

\begin{theorem}[Cunningham, Marsh \cite{Cunningham1978}]\label{cunmarsh}
	For every graph $H$, $\nu(H)$ is the minimum of
	\[ |\mathcal{U}|+\sum_{F\in \mathcal{F}}\frac{|V(F)|-1}{2} \]
	over all pairs of a set $\mathcal{U}$ of stars, and a set $\mathcal{F}$ of 2-connected factor-critical induced  subgraphs of $H$, such that every edge of $H$ belongs to a star in $\mathcal{U}$ or to $E(F)$ for some $F\in\mathcal{F}$.
\end{theorem}

We say that a graph $G$ has an \emph{odd ear-decomposition} if there exists a sequence $(P_0,P_1,\ldots,P_k)$ of an odd circuit $P_0$ and odd paths $P_1,\ldots,P_k$ of $G$ such that $\cup_{i=0}^{k}P_i=G$ and for each $i\geq 1$, only the two ends of $P_i$ belong to $P_0\cup\ldots \cup P_{i-1}$ (this type of ear-decomposition is usually called \emph{open} in the literature). 

\begin{theorem}[Lov\'asz \cite{Lovasz1972a}]\label{lovasz}
	A 2-connected graph is factor-critical if and only if it has an odd ear-decomposition.
\end{theorem}

We now prove Theorem 1:

\begin{proof}[Proof of Theorem 1]
	Let $H$ be a graph such that $G:=\cl{H}$ is $t$-perfect and $G$ does not contain $\cl{\Pi}$ nor $\cl{W_5}$ as an induced subgraph. In other words, $H$ does not contain $\Pi$ or $W_5$ as a subgraph. We will show that $\chi(G)\leq 3$ and, clearly, this will prove Theorem \ref{colcltp} (since $\cl{\Pi}$ and $\cl{W_5}$ both have chromatic number 4 (see Proposition \ref{w5}) and are not subgraphs of each other).

	We may assume that $H$ is simple: if $e$ is a multiple edge of $H$, then $e$ and any other edge parallel to $e$ define two non-adjacent vertices of $G$ which have exactly the same neighbors; it is easy to check that $\chi(\cl{H})=\chi(\cl{H-e})$, and we can conclude using induction. Obviously, we may also suppose that $H$ has no isolated vertex.
	
	Assuming that $H$ is simple, we have $\chi(G)=\gamma(H)$ and thus we need only to prove that $\gamma(H)\leq 3$ (see Section \ref{defs}). 
	
	Since $G$ is $t$-perfect, it does not contain $K_4$ and $\nu(H)=\omega(G)\leq 3$. By Theorem \ref{cunmarsh}, there exist a set $\mathcal{U}$ of stars and a set $\mathcal{F}$ of 2-connected factor-critical induced subgraphs of $H$ such that every edge of $H$ belongs to either a star in $\mathcal{U}$ or to a $E(F)$ for some $F\in\mathcal{F}$, and:
	\begin{equation}\label{CM}
	\nu(H)=|\mathcal{U}|+\sum_{F\in\mathcal{F}}\frac{|V(F)|-1}{2}.
	\end{equation}
	Clearly, if $\mathcal{F}$ contains only circuits of length 3, then $\mathcal{U}\cup \mathcal{F}$ gives an \stc of $H$ and $\gamma(H)\leq\nu(H)\leq 3$; we are done. Thus, we may assume that $\mathcal{F}$ contains a 2-connected factor-critical induced subgraph $F$ which is not a circuit of length 3. We show:
	
	\begin{equation}\label{fullC5}
		\text{\emph{$F$ has a spanning circuit of length 5.}}
	\end{equation}
	By Theorem \ref{lovasz}, $F$ has an odd ear-decomposition $(P_0,P_1,\ldots,P_k)$.
	Suppose first that $P_0$ is a circuit of length at least 5. By Corollary \ref{small}, since $G$ is $t$-perfect, $P_0$ has length exactly 5 and every edge of $H$ has at least one end in $P_0$. Hence, $P_1,\ldots,P_k$ (if they exist) must have length 1. Consequently, $P_1,\ldots,P_k$ have both their ends in $P_0$, thus $P_0$ is a spanning subgraph of $F$ and the conclusion follows.
	
	Suppose now that $P_0$ has length 3. Since $H$ is simple (by assumption) and not a circuit of length 3, we must have $k\geq 1$, and $P_1$ has length at least 3. By Corollary \ref{small}, since $G$ is $t$-perfect, $H$ does not contain an odd circuit of length at least 7. Thus, $P_1$ has length exactly 3 and  $P_0\cup P_1$ contains a circuit of length 5. Hence, as above, $P_2,\ldots,P_k$ must have length 1 and have both their ends in $P_0\cup P_1$. Therefore, any circuit of length 5 of $P_0\cup P_1$ is a spanning subgraph of $F$, as required. This proves \eqref{fullC5}.
	
	Clearly, if $\mathcal{U}\cup \mathcal{F}=\set{F}$ then $H$ is a subgraph of $K_5$ and $\gamma(H)\leq\gamma(K_5)\leq 3$ (as $\gamma$ is non-decreasing  for the subgraph relation). 
	
	Hence, we may assume that $\mathcal{U}\cup\mathcal{F}\neq\set{F}$.	This implies that $H$ has an edge $e=uv$ which is not an edge of $F$. By Corollary \ref{small}, since $G$ is $t$-perfect, $e$ must have at least one end, say $u$, in $F$. Put $u_1:=u$ and let $C:=(u_1,\ldots,u_5)$ be a spanning $C_5$ of $F$ (which exists by \eqref{fullC5}). Therefore, $v\notin V(F)$. We now show :
	
	\begin{equation}\label{unique}
	V(H)=V(F)\cup\set{v}.
	\end{equation}
	Since $|V(F)|=5$ and as $\set{F}\subsetneq\mathcal{U}\cup \mathcal{F}$, Equation \eqref{CM} shows that $\nu(H)\geq 3$. Since $G$ is $t$-perfect, $\nu(H)\leq 3$ and thus $\nu(H)=3$. Then, by Equation \eqref{CM}, $\mathcal{U}\cup\mathcal{F}=\set{F,U}$, where $U$ must be either a star or a circuit of length 3. 
	
	It is straightforward to check that each maximum matching $M$ of $H$ (that is $|M|=3$) has exactly 2 edges in the set $\eds{F}$ and exactly 1 edge in $U$ (if $U$ is a circuit of length 3, we mean that $e$ is an edge of this circuit);   this is just the complementary slackness following from Theorem \ref{cunmarsh}. In particular, the matching $\set{e,u_2u_3,u_4u_5}$ shows that $e$ belongs to $U$.
	As a consequence, if $U$ is a star then $U\se\delta_H(u_1)$ or $U\se\delta_H(v)$. And otherwise, if $U$ is a circuit of length 3, then it has exactly two vertices in $F$ (since Corollary \ref{small} shows that each edge of $H$ must have at least one end in $C$, as $G$ is $t$-perfect).
	
	Since every edge of $H$ belongs either to $E(F)$ or to $U$, this  implies overall that every edge of $H$ which is not an edge of $F$ is incident to at least one of $u_1$ or $v$. The claim \eqref{unique} then follows, as every edge has at least one end in $C$ and as $H$ has no isolated vertex by assumption.
	
	We are now ready to build an \stc of $H$ of cardinality 3, thus showing that $\gamma(H)\leq 3$ and concluding the proof. Our construction depends only on the degree of $v$ in $H$.
	Since $H$ has no isolated vertex, $d_H(v)\geq 1$. As $H$ does not contain $W_5$ (as a subgraph), we also have $d_H(v)\leq 4$.
	
	\vspace{5pt}
	\emph{{\textsc Case 1:}}
		$d_H(v)=1$. 
		If $u_2u_4\notin\eds{H}$, then we take the full stars at $u_1$, $u_3$ and $u_5$ to form a convenient \stc of $H$. Otherwise, we simply replace the full star at $u_3$ by the triangle $u_2u_3u_4$ to get an \stc of cardinality 3.
	
	\vspace{5pt}
	\emph{{\textsc Case 2:}}
		$d_H(v)=2$.
		First, suppose that the neighbors of $v$  are consecutive on $C$. Without loss of generality, we may assume that $N_G(v)=\set{u_1,u_2}$. If $u_3u_5\notin\eds{H}$, then the full stars at $u_1$, $u_2$ and $u_4$ form a convenient \stc. Otherwise, we replace the full star at $u_4$ with the triangle $u_3u_4u_5$.
		
		Now, suppose that the neighbors of $v$  are not consecutive on $C$.
		By symmetry, we may assume that $N_H(v)=\set{u_1,u_3}$. If both $u_2u_5$ and $u_2u_4$ are edges of $H$, then the triangle $u_2u_4u_5$ and the full stars at $u_1$, $u_3$ form an  \stc cover with 3 elements. Hence, suppose that one of these two edges do not belong to $H$. By symmetry again, we may assume that $u_2u_5\notin\eds{H}$. Then, the full stars at $u_1$, $u_3$ and $u_4$ form an \stc of $H$ as required.
		
	\vspace{5pt}	
	\emph{{\textsc Case 3:}}
		$d_H(v)=3$. 
		Similarly, we first suppose that the neighbors of $v$ are consecutive on $C$. Without loss of generality, we have $N_H(v)=\set{u_1,u_2,u_5}$. Since $H$ does not contain $W_5$, at most one of $u_1u_3$ and $u_1u_4$ is an edge of $H$. By symmetry, we may assume that $u_1u_4\notin\eds{H}$.
		If $u_2u_4\notin\eds{G}$, then the full stars at $u_3$ and $u_5$, and the triangle $vu_1u_2$ form an \stc of cardinality 3. Otherwise, since $H$ does not contain $\Pi$, we must have $u_1u_3\notin\eds{G}$. Therefore, the full star at $u_5$ and the triangles $vu_1u_2$ and $u_2u_3u_4$ form an \stc of cardinality 3.

		Now, suppose that the neighbors of $v$ are not consecutive on $C$. Without loss of generality, we may assume that $N_H(v)=\set{u_1,u_3,u_4}$. Since $H$ does not contain $\Pi$, we have $u_2u_5\notin\eds{H}$. Then, the full stars at $u_1$, $u_3$, $u_4$ form an \stc of cardinality 3.
	
	\vspace{5pt}	
	\emph{{\textsc Case 4:}}
		$d_H(v)=4$.
		Without loss of generality, we may assume that $u_2$ is the vertex of $C$ which is not a neighbor of $v$. Since $H$ does not contain $\Pi$, both $u_2u_4$ and $u_2u_5$ cannot be edges of $H$. Hence, the triangle $vu_4u_5$ together with the two full stars at $u_1$ and $u_3$ form an \stc as required.	
\end{proof}

This proof directly yields a polynomial-time algorithm finding a 3-coloring or showing $\cl{\Pi}$ or $\cl{W_5}$ as an induced subgraph, for every $t$-perfect graph whose complement is a line graph.

The following lemma of Seb\H{o} is the key-argument in showing Corollary \ref{colclhp} as a consequence of Theorem \ref{colcltp}. It is stated and proved in \cite{Bruhn2012} (with the additional assumption that the graph is claw-free, yet this hypothesis is not used in the proof, hence the result holds in fact for every $h$-perfect graph):

\begin{lemma}[Seb\H{o}, Lemma 26 in \cite{Bruhn2012}]\label{sebo}
	If $G$ is an $h$-perfect graph and $\omega(G)\geq 3$, then $G$ has a stable set intersecting every  clique of maximum cardinality.
\end{lemma}

\begin{proof}[Proof of Corollary \ref{colclhp}]
	Since $\cl{\Pi}$ and $\cl{W_5}$ are $t$-perfect, Proposition \ref{fcn} shows that their fractional chromatic number is at most  3. The implication i)$\Rightarrow$ii) follows, since they are not 3-colorable.
	
	Conversely, let $G$ be an $h$-perfect graph which is the complement of a line graph and which does not contain $\cl{\Pi}$ and $\cl{W_5}$ as induced subgraphs. 
	We need only to prove that $\chi(G)\leq \ceil{\chi_f(G)}$ (since the other inequality always holds).
	
	First, suppose that $\omega(G)\leq 3$. If $G$ is bipartite, then $\chi_f(G)=2=\chi(G)$ and we are done, thus we may assume that $G$ is not bipartite. Hence, $\chi_f(G)>2$. Besides, since $G$ is $t$-perfect, Proposition \ref{fcn} implies that $\chi_f(G)\leq 3$. Therefore, $\ceil{\chi_f(G)}=3$. By Theorem \ref{colcltp}, $\chi(G)=3$ and the conclusion follows.
	
	We conclude by induction on $\omega$: if $\omega(G)\geq 4$ then, by Lemma \ref{sebo}, there exists a stable set $S$ such that $\omega(G-S)=\omega(G)-1$. It is straightforward to check that Proposition \ref{fcn} implies the equalities $\chi_f(G-S)=\omega(G-S)$ and $\chi_f(G)=\omega(G)$ (as both $G$ and $G-S$ have a clique of cardinality 3). Therefore, by induction:
	\[ \chi(G) \leq \chi(G-S)+1= \chi_f(G-S)+1=\omega(G-S)+1=\omega(G)=\chi_f(G), \]
	as required.
\end{proof}

Bruhn and Stein \cite{Bruhn2012} observed that a stable set as in Lemma \ref{sebo} can be found in polynomial-time (through the Ellipsoid method \cite{Groetschel1988}). It is not surprising that this can be carried out in a much simpler way for the special case of $h$-perfect graphs whose complement is a line graph. Indeed, this is equivalent to find either  a star or a triangle intersecting every maximum matching, in a graph satisfying conditions (i) and (ii) of Corollary \ref{small}. 
Clearly, any star or triangle in an optimal cover given by Theorem \ref{cunmarsh} is convenient (by complementary slackness). On the other hand, if such an optimal cover does not contain any star or triangle, then it is easy to find a convenient star or circuit directly.

Therefore, this and the proof of Corollary \ref{colclhp} show that for every $h$-perfect graph $G$ which is the complement of a line graph, a coloring of $G$ using $\ceil{\chi_f(G)}$ colors or an induced subgraph of $G$ isomorphic to $\cl{\Pi}$ or $\cl{W_5}$ can be found in polynomial-time.

We now show a polynomial-time algorithm testing $h$-perfection in the class of complements of line graphs. 

\begin{proposition}
	$H$-perfection can be tested in polynomial-time in the class of graphs whose complement is a line graph.
\end{proposition}
\begin{proof}
	It is well-known that deciding whether a graph is a line graph (and finding a source graph if it exists) can be done in polynomial time \cite{Beineke1970}.  
	Hence, by Corollary \ref{small}, we need only testing whether an input graph $H$ satisfies the following two conditions: (i) each edge of $H$ has at least one end in each circuit of length 5, and (ii) every odd circuit of $H$ has length at most 5.
	
	We start with testing whether $H$ contains an odd circuit of length at least 5. This can be done in polynomial-time using an algorithm of Trotter \cite{trotter1977line}.
	Obviously, if $H$ has no odd circuit of length at least 5, then conditions (i) and (ii) are satisfied, thus $\cl{H}$ is $h$-perfect and we are done. On the other hand, if the algorithm of \cite{trotter1977line} shows an odd circuit of length at least 7, then (ii) is false and therefore $\cl{H}$ is not $h$-perfect.
	
	Hence we may assume that $H$ contains a circuit $C$ of length exactly 5. Since (i) can be tested in polynomial time (obviously), we may also assume that it holds. 
	
	We show that these additional assumptions ensure that $H$ cannot have a circuit of length greater than 10. 	
	Indeed, let $D$ be a circuit of $H$ of length $N\geq 6$. 
	Since $|V(D)\sm V(C)|\geq N-5$, and as (i) holds (by assumption), $V(D)\sm V(C)$ is a stable set of $D$.
	On the other hand, it is straightforward to check that the cardinality of a largest stable set of the cycle on $N$ vertices is $\floor{N/2}$. Hence, $N-5\leq\floor{N/2}$, which implies $N\leq 10$.
	
	Therefore, checking (ii) only requires testing whether $H$ contains $C_7$ or $C_9$ as a subgraph. Clearly, this can be carried out in polynomial-time.
\end{proof}
 
We end this section by showing Theorem \ref{p6free} as a direct corollary of the following result, due to Randerath, Schiermeyer and Tewes. The Mycielski-Gr\"otzsch graph is showed in Figure \ref{mg}.

\begin{figure}
	\centering
	\begin{tikzpicture}[scale=2]
	
	\tikzset{vertex/.style={draw,circle,scale=0.3}}

	\node[vertex] (v00) at (0,1) {};
	
	\draw (v00) node[above, scale=0.7]{$u$};
	
	\node[vertex] (v11) at (-1,.5) {};
	\node[vertex] (v12) at (-.5,.5) {};
	\node[vertex] (v13) at (0,.5) {};
	\node[vertex] (v14) at (.5,.5) {};
	\node[vertex] (v15) at (1,.5) {};
	
	\node[vertex] (v21) at (-1,0) {};
	\node[vertex] (v22) at (-.5,0) {};
	\node[vertex] (v23) at (0,0) {};
	\node[vertex] (v24) at (.5,0) {};
	\node[vertex] (v25) at (1,0) {};
	
	\draw (v00) edge (v11)
	edge (v12)
	edge (v13)
	edge (v14)
	edge (v15);
	
	\draw (v21)--(v22)--(v23)--(v24)--(v25);
	
	\draw (v11)--(v22)--(v13)--(v24)--(v15)--(v21)
	--(v12)--(v23)--(v14)--(v25)--(v11);
	
	\draw (v21) to[bend right] (v25);

	%
	
	%

	%
	%
	%
	%
	%

	\end{tikzpicture}
	\caption{the Mycielski-Gr\"otzsch graph}\label{mg}
\end{figure}
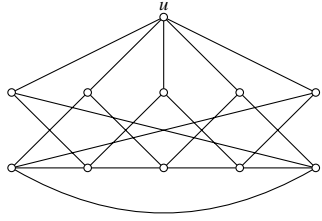

\begin{theorem}[\cite{Randerath2002}]\label{randerath}
	The Mycielski-Gr\"otzsch graph is the only \mntc triangle-free $P_6$-free graph.
\end{theorem}

Gerards and Shepherd \cite{Gerards1998} proved that \emph{for every $t$-perfect graph $G$ and for every vertex $v\in\vts{G}$ whose neighborhood is a stable set: the graph obtained by identifying $v$ and its neighbors to a single vertex is $t$-perfect} (loops and multiple edges which may arise are deleted). This implies that the Mycielski-Gr\"otzsch graph is \emph{not} $t$-perfect. Indeed, contracting every edge incident to $u$  yields $W_5$ (see Figure \ref{mg}), which is not $t$-perfect. 

Theorem \ref{p6free} easily follows by induction:

\begin{proof}[Proof of Theorem \ref{p6free}] 
	Let $G$ be a $t$-perfect $P_6$-free graph. 
	Suppose first that $\omega(G)\leq 2$, that is $G$ is triangle-free. Since the Mycielski-Gr\"otzsch graph is not $t$-perfect (see the discussion above), $G$ cannot contain it as an induced subgraph (as $t$-perfection is closed for taking induced subgraphs). Therefore, since $G$ is $P_6$-free, Theorem \ref{randerath} directly implies that $\chi(G)\leq 3$. Hence, we may assume that $\omega(G)\geq 3$. By Lemma \ref{sebo}, there exists a stable set $S$ of $G$ such that $\omega(G-S)=\omega(G)-1$. 
	 By induction, $\chi(G-S)\leq \omega(G-S)+1$, and thus $\chi(G)\leq \chi(G-S)+1\leq \omega(G)+1$, as stated (this is the same argument as in the proof of Corollary \ref{colclhp}). 
	
\end{proof}

\section{Final remarks}\label{conclusion}
We proved that  $\cl{\Pi}$ and $\cl{W_5}$ are the only minimal obstructions to 3-colorability for $t$-perfect graphs which are complements of  line graphs. We do not know whether there are more obstructions to 3-colorability for \emph{every} $t$-perfect graph, nor if there exists a 5-chromatic $t$-perfect graph. 
Besides, the computational complexity of determining the chromatic number of a $t$-perfect graph is still open.

Seb\H{o} conjectures that \emph{every $t$-perfect triangle-free graph is 3-colorable} (Conjecture 25 in \cite{Bruhn2012}). By Lemma \ref{sebo}, this would imply that every $t$-perfect graph is 4-colorable. As stated by Theorem \ref{p6free}, this conjecture is true for $P_6$-free graphs.



A graph is \emph{nearly-bipartite} if each vertex has at least one neighbor on each odd circuit. It is easy to check that nearly-bipartite graphs, as $P_6$-free graphs, form a superclass of complements of line graphs. Using Theorem 1, \cite{BenchetritUlmer} proves that \emph{$\cl{\Pi}$ and $\cl{W_5}$ are the only \mntc $t$-perfect nearly-bipartite graphs}.


We prove in Appendix \ref{appen} that $\cl{\Pi}$ and $\cl{W_5}$ are the only \mntc $t$-perfect $P_6$-free graphs (Theorem \ref{p6critical}). Notice that the proof of Corollary \ref{colclhp} shows more generally the following, for every class of graphs $\mathcal{C}$ which is closed by taking induced subgraphs. Let $\mathcal{F}$ be the set of 4-critical $t$-perfect graphs in $\mathcal{C}$. For every $h$-perfect graph $G\in\mathcal{C}$: $\chi(H)=\ceil{\chi_f(H)}$ for every induced subgraph $H$ of $G$ if and only if $G$ does not contain a graph of $\mathcal{F}$ as an induced subgraph. 

Hence, the equivalence stated by Corollary \ref{colclhp} holds more generally for $t$-perfect $P_6$-free graphs. Furthermore, results of \cite{ChudnovskyGoedgebeurSchaudtEtAl2015} and \cite{Randerath2004} easily imply that an optimal coloring or an obstruction can be found in polynomial-time. Yet, it is still open whether $t$-perfection can be tested in polynomial-time in the class of $P_6$-free graphs (Bruhn and Fuchs \cite{BruhnFuchs2015} showed that this is true at least for the subclass of $P_5$-free graphs).

We determined the 4-critical $t$-perfect $P_6$-free graphs by providing, for each graph other than $\cl{\Pi}$ and $\cl{W_5}$ in the explicit list of all \mntc $P_6$-free graphs \cite{ChudnovskyGoedgebeurSchaudtEtAl2015}, a certificate of non-$t$-perfection, built independently for each graph. The same approach is used in \cite{BruhnFuchs2015} to show that every $t$-perfect $P_5$-free graph is 3-colorable. 
Alternative proofs of these results, using more general arguments (as given here in the special case of complements of line graphs), are certainly of interest since the listing approach  fails when excluding larger induced paths: for every $k\geq 7$, the set of \mntc $P_k$-free graphs is infinite \cite{ChudnovskyGoedgebeurSchaudtEtAl2015}.



\appendix
\section{Proof of Theorem \ref{p6critical}}\label{appen}

A basic polyhedral argument shows that a $K_4$-free graph $G$ is $t$-perfect if and only if all the vertices of the polytope $\mathsf{HSTAB}(G)$  are \emph{integral}, that is they have integer coordinates only ($\mathsf{HSTAB}(G)$ is defined by the non-negativity, clique and odd circuit inequalities associated to $G$; see Section \ref{defs}). 

Hence, to prove that a $K_4$-free graph $G$ is not $t$-perfect, it suffices to show a non-integral point $z\in\mathsf{HSTAB}(G)$, and  $|\vts{G}|$ inequalities $a_1^{T}x\leq b_1,\ldots,a_{|\vts{G}|}^{T}x\leq b_{|\vts{G}|}$ of non-negativity, cliques or odd circuits, such that the real vectors $a_1,\ldots,a_{|\vts{G}|}$ are linearly independent; it is well-known that such inequalities certify that $z$ is a vertex of $\mathsf{HSTAB}(G)$. 

We now use this observation and a result of \cite{ChudnovskyGoedgebeurSchaudtEtAl2015} to prove Theorem \ref{p6critical}, which states that $\cl{\Pi}$ and $\cl{W_5}$ are the only  4-critical $t$-perfect $P_6$-free graphs.


The eighty 4-critical $P_6$-free graphs are determined explicitly in \cite{ChudnovskyGoedgebeurSchaudtEtAl2015}. In the list \cite{hog}, these graphs are referred to as \emph{4-vertex-critical} $P_6$-free graphs and are given encoded in {\sc graph6} format \cite{McKay1981}; $\cl{\Pi}$ and $\cl{W_5}$ are respectively encoded as the {\sc graph6} strings \texttt{HErb\`{}yi} and \texttt{I?Becw\}Yo}. 

The $t$-perfection of $\cl{\Pi}$ and $\cl{W_5}$ is proved in Section \ref{proofs}. Furthermore, $K_4$ is a 4-critical $P_6$-free graph whose non-$t$-perfection is trivial (it is encoded as "\texttt{C\textasciitilde}" and is the first graph in the list at \cite{hog}). Hence, it only remains to show a non-integral vertex of $\mathsf{HSTAB}$ for each of the other seventy-seven 4-critical $P_6$-free graphs (they are $K_4$-free, since they are 4-critical and not $K_4$ itself).

The two Tables \ref{table} and \ref{tablec} provide, for each of these 77 graphs, the {\sc graph6} code as given in \cite{hog}, an adjacency list and a non-integral vertex of $\mathsf{HSTAB}$ (we give non-redundant adjacency lists, meaning that for each pair $ij$ with $i<j$, $ij$ is an edge of the graph if and only if  $j$ is in the list of $i$). Checking that this point is indeed a vertex of $\mathsf{HSTAB}$ is straightforward. Hence, Theorem \ref{p6critical} is proved.

Among the eighty 4-critical $P_6$-free graphs, only five of them are such that $\mathsf{HSTAB}$ has a unique non-integral vertex. In fact, they are \emph{minimally $t$-imperfect} \cite{BruhnFuchs2015,Bruhn2012}. We emphasize these 5 graphs (except $K_4$) in the Tables \ref{table} and \ref{tablec} by giving an alternative description (in parentheses) to the {\sc graph6} code; $C_n^{k}$ denotes the graph obtained from the cycle $C_n$ by adding all edges between two vertices at distance at most $k$ on $C_n$.

\begin{table}[t]
	\caption{4-critical $P_6$-free graphs}\label{table}
\begin{adjustbox}{width=1.2\textwidth,center=\textwidth}
	\begin{tabular}{|l|l|c|}
		\hline 
		{\sc graph6 code} & adjacency list & non-integral vertex of $\mathsf{HSTAB}$ \\ \hline
		\texttt{EUZw}\quad($W_5$) & $0 :\left[2, 3, 5\right], 1 :\left[3, 4, 5\right], 2 :\left[4, 5\right], 3 :\left[5\right], 4 :\left[5\right]$ & $\left[\frac{2}{5}, \frac{2}{5}, \frac{2}{5}, \frac{2}{5}, \frac{2}{5}, \frac{1}{5}\right]$ \\ \hline 
		\texttt{FEnbo} & $0 :\left[3, 4, 5\right], 1 :\left[3, 5, 6\right], 2 :\left[4, 5, 6\right], 3 :\left[4, 6\right], 4 :\left[6\right]$ & $\left[\frac{1}{3}, \frac{1}{3}, 0, \frac{1}{3}, \frac{1}{3}, \frac{2}{3}, \frac{1}{3}\right]$ \\ \hline
		\texttt{FqjRo} & $0 :\left[1, 2, 4, 5\right], 1 :\left[3, 5, 6\right], 2 :\left[4, 6\right], 3 :\left[5, 6\right], 4 :\left[6\right]$ & $\left[\frac{1}{3}, 0, \frac{1}{3}, \frac{2}{3}, \frac{1}{3}, \frac{1}{3}, \frac{1}{3}\right]$ \\ \hline
		\texttt{FQjRo} & $0 :\left[2, 4, 5\right], 1 :\left[3, 5, 6\right], 2 :\left[4, 6\right], 3 :\left[5, 6\right], 4 :\left[6\right]$ & $\left[\frac{1}{3}, 0, \frac{1}{3}, \frac{1}{3}, \frac{1}{3}, \frac{2}{3}, \frac{1}{3}\right]$ \\ \hline
		\texttt{FrjRo} & $0 :\left[1, 2, 4, 5\right], 1 :\left[3, 5, 6\right], 2 :\left[3, 4, 6\right], 3 :\left[5, 6\right], 4 :\left[6\right]$ & $\left[\frac{1}{3}, \frac{1}{3}, 0, \frac{1}{3}, \frac{2}{3}, \frac{1}{3}, \frac{1}{3}\right]$ \\ \hline
		\texttt{FrjZo}\quad($\overline{C_7}$) & $0 :\left[1, 2, 4, 5\right], 1 :\left[3, 5, 6\right], 2 :\left[3, 4, 6\right], 3 :\left[5, 6\right], 4 :\left[5, 6\right]$ & $\left[\frac{1}{3}, \frac{1}{3}, \frac{1}{3}, \frac{1}{3}, \frac{1}{3}, \frac{1}{3}, \frac{1}{3}\right]$ \\ \hline
		\texttt{FYjRo} & $0 :\left[2, 4, 5\right], 1 :\left[2, 3, 5, 6\right], 2 :\left[4, 6\right], 3 :\left[5, 6\right], 4 :\left[6\right]$ & $\left[\frac{1}{3}, 0, \frac{1}{3}, \frac{1}{3}, \frac{1}{3}, \frac{2}{3}, \frac{1}{3}\right]$ \\ \hline
		\texttt{FYnRo} & $0 :\left[2, 4, 5\right], 1 :\left[2, 3, 5, 6\right], 2 :\left[4, 6\right], 3 :\left[4, 5, 6\right], 4 :\left[6\right]$ & $\left[\frac{2}{3}, \frac{1}{3}, \frac{1}{3}, \frac{1}{3}, 0, \frac{1}{3}, \frac{1}{3}\right]$ \\ \hline
		\texttt{GCqjbc} & $0 :\left[3, 4, 5\right], 1 :\left[4, 6, 7\right], 2 :\left[5, 6, 7\right], 3 :\left[6, 7\right], 4 :\left[5\right], 6 :\left[7\right]$ & $\left[0, \frac{1}{3}, \frac{1}{3}, 0, \frac{1}{3}, \frac{2}{3}, \frac{1}{3}, \frac{1}{3}\right]$ \\ \hline
		\texttt{GCqnbc} & $0 :\left[3, 4, 5, 6\right], 1 :\left[4, 6, 7\right], 2 :\left[5, 6, 7\right], 3 :\left[6, 7\right], 4 :\left[5\right], 6 :\left[7\right]$ & $\left[0, \frac{1}{3}, \frac{1}{3}, \frac{1}{3}, \frac{2}{3}, \frac{1}{3}, \frac{1}{3}, \frac{1}{3}\right]$ \\ \hline
		\texttt{GCqnbs} & $0 :\left[3, 4, 5, 6\right], 1 :\left[4, 6, 7\right], 2 :\left[5, 6, 7\right], 3 :\left[6, 7\right], 4 :\left[5, 7\right], 6 :\left[7\right]$ & $\left[0, \frac{1}{3}, \frac{1}{3}, 0, \frac{1}{3}, \frac{2}{3}, \frac{1}{3}, \frac{1}{3}\right]$ \\ \hline
		\texttt{Gcrbds} & $0 :\left[1, 3, 4, 5, 7\right], 1 :\left[4, 5, 6\right], 2 :\left[5, 6, 7\right], 3 :\left[6, 7\right], 4 :\left[7\right], 6 :\left[7\right]$ & $\left[\frac{1}{3}, \frac{1}{3}, \frac{2}{3}, 0, \frac{1}{3}, \frac{1}{3}, 0, \frac{1}{3}\right]$ \\ \hline
		\texttt{GCrbds} & $0 :\left[3, 4, 5, 7\right], 1 :\left[4, 5, 6\right], 2 :\left[5, 6, 7\right], 3 :\left[6, 7\right], 4 :\left[7\right], 6 :\left[7\right]$ & $\left[\frac{1}{3}, \frac{1}{3}, \frac{1}{3}, \frac{1}{3}, 0, \frac{2}{3}, \frac{1}{3}, \frac{1}{3}\right]$ \\ \hline
		\texttt{GCRdrs} & $0 :\left[3, 5, 6\right], 1 :\left[4, 5, 7\right], 2 :\left[5, 6, 7\right], 3 :\left[6, 7\right], 4 :\left[6, 7\right], 6 :\left[7\right]$ & $\left[\frac{1}{3}, \frac{1}{3}, \frac{1}{3}, \frac{1}{3}, 0, \frac{2}{3}, \frac{1}{3}, \frac{1}{3}\right]$ \\ \hline
		\texttt{HCQ\`{}fp]} & $0 :\left[3, 5, 7\right], 1 :\left[4, 7, 8\right], 2 :\left[5, 6, 7\right], 3 :\left[6, 7, 8\right], 4 :\left[7, 8\right], 5 :\left[8\right], 6 :\left[8\right]$ & $\left[\frac{1}{3}, \frac{1}{3}, \frac{1}{3}, 0, \frac{1}{3}, \frac{2}{3}, 0, \frac{1}{3}, \frac{1}{3}\right]$ \\ \hline
		\texttt{HCQ\`{}fP]} & $0 :\left[3, 5, 7\right], 1 :\left[4, 7, 8\right], 2 :\left[5, 6, 7\right], 3 :\left[6, 8\right], 4 :\left[7, 8\right], 5 :\left[8\right], 6 :\left[8\right]$ & $\left[\frac{2}{3}, \frac{1}{3}, \frac{1}{3}, \frac{1}{3}, \frac{1}{3}, 0, \frac{1}{3}, \frac{1}{3}, \frac{1}{3}\right]$ \\ \hline
		\texttt{HCQ\`{}fR]} & $0 :\left[3, 5, 7, 8\right], 1 :\left[4, 7, 8\right], 2 :\left[5, 6, 7\right], 3 :\left[6, 8\right], 4 :\left[7, 8\right], 5 :\left[8\right], 6 :\left[8\right]$ & $\left[0, \frac{1}{3}, \frac{1}{3}, \frac{2}{3}, \frac{1}{3}, \frac{2}{3}, 0, \frac{1}{3}, \frac{1}{3}\right]$ \\ \hline
		\texttt{HCQ\`{}fV]} & $0 :\left[3, 5, 7, 8\right], 1 :\left[4, 7, 8\right], 2 :\left[5, 6, 7\right], 3 :\left[6, 8\right], 4 :\left[7, 8\right], 5 :\left[8\right], 6 :\left[7, 8\right]$ & $\left[0, \frac{1}{3}, \frac{1}{3}, \frac{1}{3}, \frac{1}{3}, \frac{2}{3}, \frac{1}{3}, \frac{1}{3}, \frac{1}{3}\right]$ \\ \hline
		\texttt{HCQ\`{}fX]} & $0 :\left[3, 5, 7\right], 1 :\left[4, 7, 8\right], 2 :\left[5, 6, 7\right], 3 :\left[6, 8\right], 4 :\left[7, 8\right], 5 :\left[7, 8\right], 6 :\left[8\right]$ & $\left[\frac{1}{2}, \frac{1}{4}, \frac{1}{2}, \frac{1}{4}, \frac{1}{4}, 0, \frac{1}{4}, \frac{1}{2}, \frac{1}{2}\right]$ \\ \hline
		\texttt{HCq\`{}qjy} & $0 :\left[3, 4, 5, 8\right], 1 :\left[4, 7, 8\right], 2 :\left[5, 6, 8\right], 3 :\left[6, 7, 8\right], 4 :\left[6, 8\right], 5 :\left[7\right], 6 :\left[8\right]$ & $\left[\frac{1}{3}, \frac{1}{3}, \frac{1}{3}, \frac{1}{3}, \frac{1}{3}, \frac{1}{3}, 0, \frac{2}{3}, \frac{1}{3}\right]$ \\ \hline
		\texttt{HCq\`{}v\`{}]} & $0 :\left[3, 4, 5, 7\right], 1 :\left[4, 7, 8\right], 2 :\left[5, 6, 7\right], 3 :\left[6, 7, 8\right], 4 :\left[6, 8\right], 5 :\left[8\right], 6 :\left[8\right]$ & $\left[0, \frac{1}{3}, \frac{1}{3}, \frac{1}{3}, \frac{1}{3}, \frac{1}{3}, \frac{1}{3}, \frac{2}{3}, \frac{1}{3}\right]$ \\ \hline
		\texttt{HCrb\`{}qi} & $0 :\left[3, 4, 5, 8\right], 1 :\left[4, 5, 6\right], 2 :\left[5, 6, 7, 8\right], 3 :\left[6, 7\right], 4 :\left[7, 8\right], 6 :\left[8\right]$ & $\left[0, \frac{1}{3}, \frac{1}{3}, \frac{2}{3}, \frac{1}{3}, \frac{2}{3}, \frac{1}{3}, \frac{1}{3}, \frac{1}{3}\right]$ \\ \hline
		\texttt{HErb\`{}qi} & $0 :\left[3, 4, 5, 8\right], 1 :\left[3, 4, 5, 6\right], 2 :\left[5, 6, 7, 8\right], 3 :\left[6, 7\right], 4 :\left[7, 8\right], 6 :\left[8\right]$ & $\left[\frac{2}{5}, \frac{2}{5}, \frac{1}{5}, 0, \frac{1}{5}, \frac{2}{5}, \frac{2}{5}, \frac{4}{5}, \frac{2}{5}\right]$ \\ \hline
		\texttt{HKq\`{}v\`{}]} & $0 :\left[3, 4, 5, 7\right], 1 :\left[2, 4, 7, 8\right], 2 :\left[5, 6, 7\right], 3 :\left[6, 7, 8\right], 4 :\left[6, 8\right], 5 :\left[8\right], 6 :\left[8\right]$ & $\left[0, \frac{1}{3}, 0, \frac{1}{3}, \frac{1}{3}, \frac{2}{3}, \frac{1}{3}, \frac{2}{3}, \frac{1}{3}\right]$ \\ \hline
		\texttt{HOq\`{}v\`{}]} & $0 :\left[2, 4, 5, 7\right], 1 :\left[4, 7, 8\right], 2 :\left[5, 6, 7\right], 3 :\left[6, 7, 8\right], 4 :\left[6, 8\right], 5 :\left[8\right], 6 :\left[8\right]$ & $\left[0, \frac{2}{5}, \frac{1}{5}, \frac{2}{5}, \frac{2}{5}, \frac{4}{5}, \frac{2}{5}, \frac{2}{5}, \frac{1}{5}\right]$ \\ \hline
		\texttt{H?q\`{}qjy} & $0 :\left[4, 5, 8\right], 1 :\left[4, 7, 8\right], 2 :\left[5, 6, 8\right], 3 :\left[6, 7, 8\right], 4 :\left[6, 8\right], 5 :\left[7\right], 6 :\left[8\right]$ & $\left[\frac{1}{3}, \frac{1}{3}, \frac{1}{3}, \frac{1}{3}, \frac{1}{3}, \frac{2}{3}, 0, \frac{1}{3}, \frac{1}{3}\right]$ \\ \hline
		\texttt{H?q\`{}v\`{}]} & $0 :\left[4, 5, 7\right], 1 :\left[4, 7, 8\right], 2 :\left[5, 6, 7\right], 3 :\left[6, 7, 8\right], 4 :\left[6, 8\right], 5 :\left[8\right], 6 :\left[8\right]$ & $\left[0, \frac{1}{3}, 0, \frac{1}{3}, \frac{1}{3}, \frac{2}{3}, \frac{1}{3}, \frac{2}{3}, \frac{1}{3}\right]$ \\ \hline
		\texttt{H?q\`{}vh]} & $0 :\left[4, 5, 7\right], 1 :\left[4, 7, 8\right], 2 :\left[5, 6, 7\right], 3 :\left[6, 7, 8\right], 4 :\left[6, 8\right], 5 :\left[7, 8\right], 6 :\left[8\right]$ & $\left[0, \frac{1}{3}, 0, \frac{1}{3}, \frac{1}{3}, \frac{1}{3}, \frac{1}{3}, \frac{2}{3}, \frac{1}{3}\right]$ \\ \hline
		\texttt{HSq\`{}v\`{}]} & $0 :\left[2, 3, 4, 5, 7\right], 1 :\left[4, 7, 8\right], 2 :\left[5, 6, 7\right], 3 :\left[6, 7, 8\right], 4 :\left[6, 8\right], 5 :\left[8\right], 6 :\left[8\right]$ & $\left[0, \frac{2}{5}, \frac{1}{5}, \frac{2}{5}, \frac{2}{5}, \frac{4}{5}, \frac{2}{5}, \frac{2}{5}, \frac{1}{5}\right]$ \\ \hline
		\texttt{IAbBV\_\{d\_} & $0 :\left[4, 5, 7, 9\right], 1 :\left[3, 5, 6, 7\right], 2 :\left[6, 7, 8\right], 3 :\left[7, 8, 9\right], 4 :\left[6, 8\right], 5 :\left[8, 9\right], 6 :\left[9\right]$ & $\left[\frac{1}{2}, \frac{1}{3}, \frac{1}{3}, \frac{1}{3}, \frac{1}{3}, \frac{1}{3}, \frac{1}{2}, \frac{1}{3}, \frac{1}{2}, 0\right]$ \\ \hline
		\texttt{I?bBV\_\{d\_} & $0 :\left[4, 5, 7, 9\right], 1 :\left[5, 6, 7\right], 2 :\left[6, 7, 8\right], 3 :\left[7, 8, 9\right], 4 :\left[6, 8\right], 5 :\left[8, 9\right], 6 :\left[9\right]$ & $\left[0, \frac{1}{3}, \frac{1}{3}, \frac{1}{3}, \frac{2}{3}, \frac{1}{3}, \frac{1}{3}, \frac{2}{3}, \frac{1}{3}, \frac{1}{3}\right]$ \\ \hline
		\texttt{ICOebGmeO} & $0 :\left[3, 6, 9\right], 1 :\left[4, 6, 7\right], 2 :\left[5, 7, 8\right], 3 :\left[6, 9\right], 4 :\left[8, 9\right], 5 :\left[7, 8\right], 6 :\left[8\right], 7 :\left[9\right]$ & $\left[0, \frac{1}{3}, 0, \frac{2}{3}, \frac{1}{3}, \frac{2}{3}, \frac{1}{3}, \frac{1}{3}, \frac{1}{3}, \frac{1}{3}\right]$ \\ \hline
		\texttt{ICOebKmeO} & $0 :\left[3, 6, 9\right], 1 :\left[4, 6, 7\right], 2 :\left[5, 7, 8\right], 3 :\left[6, 9\right], 4 :\left[8, 9\right], 5 :\left[7, 8\right], 6 :\left[7, 8\right], 7 :\left[9\right]$ & $\left[0, \frac{1}{3}, 0, \frac{2}{3}, \frac{1}{3}, \frac{2}{3}, \frac{1}{3}, \frac{1}{3}, \frac{1}{3}, \frac{1}{3}\right]$ \\ \hline
		\texttt{ICOebKmeW} & $0 :\left[3, 6, 9\right], 1 :\left[4, 6, 7\right], 2 :\left[5, 7, 8\right], 3 :\left[6, 9\right], 4 :\left[8, 9\right], 5 :\left[7, 8\right], 6 :\left[7, 8\right], 7 :\left[9\right], 8 :\left[9\right]$ & $\left[0, \frac{2}{3}, \frac{1}{3}, 1, \frac{1}{3}, \frac{1}{3}, 0, \frac{1}{3}, \frac{1}{3}, 0\right]$ \\ \hline
		\texttt{ICQfPxsmG} & $0 :\left[3, 5, 6, 9\right], 1 :\left[4, 6, 8\right], 2 :\left[5, 6, 7, 8, 9\right], 3 :\left[7, 8, 9\right], 4 :\left[6, 7, 9\right], 5 :\left[7, 8\right], 8 :\left[9\right]$ & $\left[\frac{1}{3}, \frac{1}{3}, 0, 0, \frac{1}{3}, \frac{2}{3}, \frac{1}{3}, \frac{1}{3}, \frac{1}{3}, \frac{1}{3}\right]$ \\ \hline
		\texttt{IcQf@xsmG} & $0 :\left[1, 3, 5, 6, 9\right], 1 :\left[4, 6, 8\right], 2 :\left[5, 6, 7, 8, 9\right], 3 :\left[7, 8, 9\right], 4 :\left[7, 9\right], 5 :\left[7, 8\right], 8 :\left[9\right]$ & $\left[0, \frac{2}{3}, \frac{1}{3}, \frac{1}{3}, \frac{1}{3}, \frac{1}{3}, \frac{1}{3}, \frac{1}{3}, \frac{1}{3}, 0\right]$ \\ \hline
		\texttt{ICQf\`{}xsmG} & $0 :\left[3, 5, 6, 9\right], 1 :\left[4, 6, 8\right], 2 :\left[5, 6, 7, 8, 9\right], 3 :\left[6, 7, 8, 9\right], 4 :\left[7, 9\right], 5 :\left[7, 8\right], 8 :\left[9\right]$ & $\left[0, 0, \frac{1}{3}, \frac{1}{3}, \frac{2}{3}, \frac{1}{3}, \frac{2}{3}, \frac{1}{3}, \frac{1}{3}, \frac{1}{3}\right]$ \\ \hline
		\texttt{ICQf@xsmG} & $0 :\left[3, 5, 6, 9\right], 1 :\left[4, 6, 8\right], 2 :\left[5, 6, 7, 8, 9\right], 3 :\left[7, 8, 9\right], 4 :\left[7, 9\right], 5 :\left[7, 8\right], 8 :\left[9\right]$ & $\left[\frac{4}{5}, \frac{2}{5}, \frac{2}{5}, 0, \frac{2}{5}, \frac{1}{5}, \frac{1}{5}, \frac{2}{5}, \frac{2}{5}, \frac{1}{5}\right]$ \\ \hline
		\texttt{IcQn@xsmG} & $0 :\left[1, 3, 5, 6, 9\right], 1 :\left[4, 6, 8\right], 2 :\left[5, 6, 7, 8, 9\right], 3 :\left[7, 8, 9\right], 4 :\left[5, 7, 9\right], 5 :\left[7, 8\right], 8 :\left[9\right]$ & $\left[0, \frac{2}{3}, \frac{1}{3}, \frac{2}{3}, \frac{1}{3}, \frac{1}{3}, \frac{1}{3}, \frac{1}{3}, 0, \frac{1}{3}\right]$ \\ \hline
		\texttt{ICq\`{}qjoq\_} & $0 :\left[3, 4, 5, 8, 9\right], 1 :\left[4, 7, 8, 9\right], 2 :\left[5, 6, 8\right], 3 :\left[6, 7, 8\right], 4 :\left[6, 9\right], 5 :\left[7\right], 6 :\left[9\right]$ & $\left[0, \frac{1}{3}, \frac{1}{3}, \frac{1}{3}, 0, \frac{2}{3}, \frac{1}{3}, \frac{1}{3}, \frac{1}{3}, \frac{2}{3}\right]$ \\ \hline
		\texttt{I\_\`{}F?|wlG} & $0 :\left[1, 4, 6, 9\right], 1 :\left[5, 6, 8\right], 2 :\left[6, 8, 9\right], 3 :\left[7, 8, 9\right], 4 :\left[7, 8\right], 5 :\left[7, 9\right], 6 :\left[7\right], 8 :\left[9\right]$ & $\left[0, \frac{1}{3}, \frac{1}{3}, \frac{1}{3}, \frac{1}{3}, 0, \frac{2}{3}, \frac{1}{3}, \frac{1}{3}, \frac{1}{3}\right]$ \\ \hline
		\texttt{I?\`{}F?|wlG} & $0 :\left[4, 6, 9\right], 1 :\left[5, 6, 8\right], 2 :\left[6, 8, 9\right], 3 :\left[7, 8, 9\right], 4 :\left[7, 8\right], 5 :\left[7, 9\right], 6 :\left[7\right], 8 :\left[9\right]$ & $\left[\frac{3}{7}, \frac{3}{7}, \frac{3}{7}, \frac{3}{7}, \frac{3}{7}, 0, \frac{3}{7}, \frac{3}{7}, \frac{2}{7}, \frac{2}{7}\right]$ \\ \hline
		\texttt{I\_\`{}N?|wlG} & $0 :\left[1, 4, 6, 9\right], 1 :\left[5, 6, 8\right], 2 :\left[6, 8, 9\right], 3 :\left[7, 8, 9\right], 4 :\left[5, 7, 8\right], 5 :\left[7, 9\right], 6 :\left[7\right], 8 :\left[9\right]$ & $\left[0, \frac{1}{3}, \frac{1}{3}, \frac{1}{3}, 0, \frac{1}{3}, \frac{2}{3}, \frac{1}{3}, \frac{1}{3}, \frac{1}{3}\right]$ \\ \hline
		\texttt{IObecw\}Yo} & $0 :\left[2, 4, 5, 6, 7\right], 1 :\left[5, 6, 9\right], 2 :\left[5, 8, 9\right], 3 :\left[6, 7, 8\right], 4 :\left[7, 8, 9\right], 5 :\left[7, 8\right], 6 :\left[8, 9\right], 7 :\left[9\right]$ & $\left[\frac{1}{3}, \frac{2}{3}, \frac{1}{3}, \frac{2}{3}, \frac{1}{3}, \frac{1}{3}, 0, \frac{1}{3}, 0, \frac{1}{3}\right]$ \\ \hline
		\texttt{IOBecw\}Yo} & $0 :\left[2, 5, 6, 7\right], 1 :\left[5, 6, 9\right], 2 :\left[5, 8, 9\right], 3 :\left[6, 7, 8\right], 4 :\left[7, 8, 9\right], 5 :\left[7, 8\right], 6 :\left[8, 9\right], 7 :\left[9\right]$ & $\left[\frac{1}{3}, \frac{2}{3}, \frac{1}{3}, \frac{1}{3}, \frac{2}{3}, \frac{1}{3}, \frac{1}{3}, \frac{1}{3}, \frac{1}{3}, 0\right]$ \\ \hline
		\texttt{I?q\`{}qjoq\_} & $0 :\left[4, 5, 8, 9\right], 1 :\left[4, 7, 8, 9\right], 2 :\left[5, 6, 8\right], 3 :\left[6, 7, 8\right], 4 :\left[6, 9\right], 5 :\left[7\right], 6 :\left[9\right]$ & $\left[0, \frac{1}{3}, 0, \frac{1}{3}, \frac{1}{3}, \frac{1}{3}, \frac{1}{3}, \frac{2}{3}, \frac{2}{3}, \frac{1}{3}\right]$ \\ \hline
		\texttt{I?q\`{}qjqq\_} & $0 :\left[4, 5, 8, 9\right], 1 :\left[4, 7, 8, 9\right], 2 :\left[5, 6, 8\right], 3 :\left[6, 7, 8\right], 4 :\left[6, 9\right], 5 :\left[7\right], 6 :\left[8, 9\right]$ & $\left[0, \frac{1}{3}, \frac{1}{3}, \frac{1}{3}, \frac{2}{3}, \frac{2}{3}, \frac{1}{3}, \frac{1}{3}, \frac{1}{3}, 0\right]$ \\ \hline
		\texttt{IsPubGmeO} & $0 :\left[1, 2, 3, 6, 9\right], 1 :\left[4, 5, 6, 7\right], 2 :\left[5, 7, 8\right], 3 :\left[5, 6, 9\right], 4 :\left[8, 9\right], 5 :\left[7, 8\right], 6 :\left[8\right], 7 :\left[9\right]$ & $\left[0, \frac{1}{3}, \frac{2}{3}, \frac{2}{3}, \frac{1}{3}, 0, \frac{1}{3}, \frac{1}{3}, \frac{1}{3}, \frac{1}{3}\right]$ \\ \hline
		\texttt{IsPubGmeW} & $0 :\left[1, 2, 3, 6, 9\right], 1 :\left[4, 5, 6, 7\right], 2 :\left[5, 7, 8\right], 3 :\left[5, 6, 9\right], 4 :\left[8, 9\right], 5 :\left[7, 8\right], 6 :\left[8\right], 7 :\left[9\right], 8 :\left[9\right]$ & $\left[0, \frac{1}{3}, \frac{1}{3}, \frac{1}{3}, 0, \frac{1}{3}, \frac{2}{3}, \frac{1}{3}, \frac{1}{3}, 0\right]$ \\ \hline
		\texttt{IsPubKmeO} & $0 :\left[1, 2, 3, 6, 9\right], 1 :\left[4, 5, 6, 7\right], 2 :\left[5, 7, 8\right], 3 :\left[5, 6, 9\right], 4 :\left[8, 9\right], 5 :\left[7, 8\right], 6 :\left[7, 8\right], 7 :\left[9\right]$ & $\left[0, \frac{2}{5}, \frac{2}{5}, \frac{4}{5}, \frac{2}{5}, \frac{1}{5}, \frac{1}{5}, \frac{2}{5}, \frac{2}{5}, \frac{1}{5}\right]$ \\ \hline

	\end{tabular}
\end{adjustbox}
\end{table}

\begin{table}
	
	\caption{4-critical $P_6$-free graphs (continued)}\label{tablec}
	\renewcommand{\arraystretch}{1.9}
	\begin{adjustbox}{width=1.7\textwidth,center=\textwidth}
		\begin{tabular}{|l|l|c|}	
			\hline 
			{\sc graph6 code} & adjacency list & non-integral vertex of $\mathsf{HSTAB}$ \\ \hline
			
			\texttt{IsPubKmeW} & $0 :\left[1, 2, 3, 6, 9\right], 1 :\left[4, 5, 6, 7\right], 2 :\left[5, 7, 8\right], 3 :\left[5, 6, 9\right], 4 :\left[8, 9\right], 5 :\left[7, 8\right], 6 :\left[7, 8\right], 7 :\left[9\right], 8 :\left[9\right]$ & $\left[0, \frac{1}{3}, \frac{1}{3}, \frac{2}{3}, \frac{2}{3}, \frac{1}{3}, \frac{1}{3}, \frac{1}{3}, \frac{1}{3}, 0\right]$ \\ \hline
			\texttt{IsSubGmeO} & $0 :\left[1, 2, 3, 6, 9\right], 1 :\left[4, 6, 7\right], 2 :\left[5, 7, 8\right], 3 :\left[4, 5, 6, 9\right], 4 :\left[8, 9\right], 5 :\left[7, 8\right], 6 :\left[8\right], 7 :\left[9\right]$ & $\left[0, \frac{1}{3}, \frac{1}{3}, \frac{1}{3}, \frac{1}{3}, \frac{1}{3}, \frac{2}{3}, \frac{1}{3}, \frac{1}{3}, \frac{1}{3}\right]$ \\ \hline
			\texttt{IsSubKmeO}  & $0 :\left[1, 2, 3, 6, 9\right], 1 :\left[4, 6, 7\right], 2 :\left[5, 7, 8\right], 3 :\left[4, 5, 6, 9\right], 4 :\left[8, 9\right], 5 :\left[7, 8\right], 6 :\left[7, 8\right], 7 :\left[9\right]$ & $\left[\frac{1}{3}, \frac{2}{3}, \frac{1}{3}, \frac{1}{3}, \frac{1}{3}, \frac{1}{3}, 0, \frac{1}{3}, \frac{1}{3}, 0\right]$ \\ \hline
			\texttt{IsSubKmeW} & $0 :\left[1, 2, 3, 6, 9\right], 1 :\left[4, 6, 7\right], 2 :\left[5, 7, 8\right], 3 :\left[4, 5, 6, 9\right], 4 :\left[8, 9\right], 5 :\left[7, 8\right], 6 :\left[7, 8\right], 7 :\left[9\right], 8 :\left[9\right]$ & $\left[\frac{1}{3}, \frac{1}{3}, \frac{1}{3}, 0, \frac{2}{3}, \frac{1}{3}, 0, \frac{1}{3}, \frac{1}{3}, 0\right]$ \\ \hline
			\texttt{Is\textbackslash ubGmeO} & $0 :\left[1, 2, 3, 6, 9\right], 1 :\left[4, 5, 6, 7\right], 2 :\left[4, 5, 7, 8\right], 3 :\left[4, 5, 6, 9\right], 4 :\left[8, 9\right], 5 :\left[7, 8\right], 6 :\left[8\right], 7 :\left[9\right]$ & $\left[0, \frac{1}{3}, \frac{1}{3}, \frac{1}{3}, \frac{1}{3}, \frac{1}{3}, \frac{2}{3}, \frac{1}{3}, \frac{1}{3}, \frac{1}{3}\right]$ \\ \hline
			\texttt{Is\textbackslash ubKmeO} & $0 :\left[1, 2, 3, 6, 9\right], 1 :\left[4, 5, 6, 7\right], 2 :\left[4, 5, 7, 8\right], 3 :\left[4, 5, 6, 9\right], 4 :\left[8, 9\right], 5 :\left[7, 8\right], 6 :\left[7, 8\right], 7 :\left[9\right]$ & $\left[\frac{1}{3}, \frac{1}{3}, \frac{1}{3}, 0, \frac{1}{3}, \frac{1}{3}, \frac{1}{3}, \frac{1}{3}, \frac{1}{3}, \frac{2}{3}\right]$ \\ \hline
			\texttt{Is\textbackslash ubKmeW}\quad($\overline{C_{10}^{2}}$) & $0 :\left[1, 2, 3, 6, 9\right], 1 :\left[4, 5, 6, 7\right], 2 :\left[4, 5, 7, 8\right], 3 :\left[4, 5, 6, 9\right], 4 :\left[8, 9\right], 5 :\left[7, 8\right], 6 :\left[7, 8\right], 7 :\left[9\right], 8 :\left[9\right]$ & $\left[\frac{1}{3}, \frac{1}{3}, \frac{1}{3}, \frac{1}{3}, \frac{1}{3}, \frac{1}{3}, \frac{1}{3}, \frac{1}{3}, \frac{1}{3}, \frac{1}{3}\right]$ \\ \hline
			\texttt{IsWubGmeO} & $0 :\left[1, 2, 3, 6, 9\right], 1 :\left[4, 6, 7\right], 2 :\left[4, 5, 7, 8\right], 3 :\left[5, 6, 9\right], 4 :\left[8, 9\right], 5 :\left[7, 8\right], 6 :\left[8\right], 7 :\left[9\right]$ & $\left[0, \frac{2}{3}, \frac{1}{3}, 0, 0, \frac{1}{3}, \frac{1}{3}, \frac{1}{3}, \frac{1}{3}, \frac{2}{3}\right]$ \\ \hline
			\texttt{IsXubGmeO} & $0 :\left[1, 2, 3, 6, 9\right], 1 :\left[4, 5, 6, 7\right], 2 :\left[4, 5, 7, 8\right], 3 :\left[5, 6, 9\right], 4 :\left[8, 9\right], 5 :\left[7, 8\right], 6 :\left[8\right], 7 :\left[9\right]$ & $\left[0, \frac{1}{3}, \frac{1}{3}, \frac{2}{3}, \frac{1}{3}, \frac{1}{3}, \frac{1}{3}, \frac{1}{3}, \frac{1}{3}, \frac{1}{3}\right]$ \\ \hline
			\texttt{IsXubGmeW} & $0 :\left[1, 2, 3, 6, 9\right], 1 :\left[4, 5, 6, 7\right], 2 :\left[4, 5, 7, 8\right], 3 :\left[5, 6, 9\right], 4 :\left[8, 9\right], 5 :\left[7, 8\right], 6 :\left[8\right], 7 :\left[9\right], 8 :\left[9\right]$ & $\left[0, \frac{1}{3}, \frac{1}{3}, \frac{1}{3}, \frac{1}{3}, \frac{1}{3}, \frac{2}{3}, \frac{1}{3}, \frac{1}{3}, \frac{1}{3}\right]$ \\ \hline
			\texttt{IsXubKmeO} & $0 :\left[1, 2, 3, 6, 9\right], 1 :\left[4, 5, 6, 7\right], 2 :\left[4, 5, 7, 8\right], 3 :\left[5, 6, 9\right], 4 :\left[8, 9\right], 5 :\left[7, 8\right], 6 :\left[7, 8\right], 7 :\left[9\right]$ & $\left[0, \frac{1}{3}, \frac{1}{3}, \frac{1}{3}, \frac{1}{3}, \frac{1}{3}, \frac{1}{3}, 0, \frac{1}{3}, \frac{2}{3}\right]$ \\ \hline
			\texttt{IWq\`{}qjoq\_} & $0 :\left[2, 4, 5, 8, 9\right], 1 :\left[2, 4, 7, 8, 9\right], 2 :\left[5, 6, 8\right], 3 :\left[6, 7, 8\right], 4 :\left[6, 9\right], 5 :\left[7\right], 6 :\left[9\right]$ & $\left[0, \frac{1}{3}, \frac{1}{3}, \frac{1}{3}, \frac{2}{3}, \frac{2}{3}, \frac{1}{3}, \frac{1}{3}, \frac{1}{3}, 0\right]$ \\ \hline
			\texttt{J?BD@g]Qvo?} & $0 :\left[5, 6, 10\right], 1 :\left[5, 9, 10\right], 2 :\left[6, 7, 10\right], 3 :\left[7, 8, 10\right], 4 :\left[8, 9, 10\right], 5 :\left[7, 8\right], 6 :\left[8, 9\right], 7 :\left[9\right]$ & $\left[0, 0, \frac{1}{3}, \frac{1}{3}, \frac{1}{3}, \frac{2}{3}, \frac{1}{3}, \frac{1}{3}, \frac{1}{3}, \frac{1}{3}, \frac{2}{3}\right]$ \\ \hline
			\texttt{J?\`{}DCcw\textasciicircum Fh?} & $0 :\left[4, 6, 7, 10\right], 1 :\left[5, 9, 10\right], 2 :\left[6, 8, 9, 10\right], 3 :\left[7, 8, 9, 10\right], 4 :\left[8, 9\right], 5 :\left[9, 10\right], 6 :\left[7\right], 8 :\left[10\right]$ & $\left[0, 0, \frac{1}{3}, \frac{1}{3}, \frac{1}{3}, \frac{2}{3}, \frac{2}{3}, \frac{1}{3}, \frac{1}{3}, \frac{1}{3}, \frac{1}{3}\right]$ \\ \hline
			\texttt{J?\`{}DCcw\textasciitilde Fh?} & $0 :\left[4, 6, 7, 9, 10\right], 1 :\left[5, 9, 10\right], 2 :\left[6, 8, 9, 10\right], 3 :\left[7, 8, 9, 10\right], 4 :\left[8, 9\right], 5 :\left[9, 10\right], 6 :\left[7\right], 8 :\left[10\right]$ & $\left[0, 0, \frac{1}{3}, \frac{1}{3}, \frac{1}{3}, \frac{2}{3}, \frac{1}{3}, \frac{2}{3}, \frac{1}{3}, \frac{1}{3}, \frac{1}{3}\right]$ \\ \hline
			\texttt{K?ABA\_YIV\{{}\textasciicircum c} & $0 :\left[5, 10, 11\right], 1 :\left[6, 7, 10, 11\right], 2 :\left[6, 9, 10, 11\right], 3 :\left[7, 8, 10, 11\right], 4 :\left[8, 9, 10, 11\right], 5 :\left[10, 11\right], 6 :\left[8, 10\right], 7 :\left[9\right], 8 :\left[11\right]$ & $\left[0, 0, \frac{1}{3}, \frac{1}{3}, \frac{1}{3}, \frac{2}{3}, \frac{2}{3}, \frac{2}{3}, \frac{1}{3}, \frac{1}{3}, 0, \frac{1}{3}\right]$ \\ \hline
			\texttt{L?\`{}DF\`{}Y\}DwFwFs} & $0 :\left[4, 6, 7, 9, 10\right], 1 :\left[5, 7, 8, 9\right], 2 :\left[6, 7, 9, 10, 11\right], 3 :\left[7, 8, 9, 10, 11, 12\right], 4 :\left[8, 9, 10, 11, 12\right], 5 :\left[10, 11, 12\right], 6 :\left[8, 11, 12\right], 7 :\left[11, 12\right], 9 :\left[12\right]$ & $\left[0, \frac{2}{3}, \frac{1}{3}, \frac{1}{3}, \frac{1}{3}, 0, \frac{1}{3}, \frac{1}{3}, \frac{1}{3}, 0, \frac{2}{3}, \frac{1}{3}, 0\right]$ \\ \hline
			\texttt{L?\`{}FBQclQt\textasciitilde C\textasciitilde ?} & $0 :\left[4, 6, 8, 9, 11, 12\right], 1 :\left[5, 6, 7, 10, 11, 12\right], 2 :\left[6, 7, 8, 9, 11, 12\right], 3 :\left[9, 10, 11, 12\right], 4 :\left[7, 10, 11, 12\right], 5 :\left[8, 9, 12\right], 6 :\left[10\right], 7 :\left[9\right], 8 :\left[10, 11\right], 9 :\left[10\right]$ & $\left[0, 0, 0, \frac{1}{3}, \frac{1}{3}, \frac{1}{3}, \frac{2}{3}, \frac{2}{3}, \frac{2}{3}, \frac{1}{3}, \frac{1}{3}, \frac{1}{3}, \frac{1}{3}\right]$ \\ \hline
			\texttt{L?\`{}FBqdlQt~C~?} & $0 :\left[4, 6, 8, 9, 11, 12\right], 1 :\left[5, 6, 7, 10, 11, 12\right], 2 :\left[6, 7, 8, 9, 11, 12\right], 3 :\left[7, 9, 10, 11, 12\right], 4 :\left[7, 10, 11, 12\right], 5 :\left[8, 9, 12\right], 6 :\left[10\right], 7 :\left[8, 9\right], 8 :\left[10, 11\right], 9 :\left[10\right]$ & $\left[0, \frac{2}{5}, \frac{2}{5}, \frac{2}{5}, \frac{4}{5}, \frac{2}{5}, \frac{2}{5}, 0, \frac{2}{5}, \frac{2}{5}, \frac{1}{5}, \frac{1}{5}, \frac{1}{5}\right]$ \\ \hline
			\texttt{L?\`{}FEb\{{}Fdk]aNo} & $0 :\left[4, 6, 7, 8, 10, 11\right], 1 :\left[5, 6, 7, 8, 11\right], 2 :\left[6, 8, 10, 11, 12\right], 3 :\left[7, 8, 9, 10, 11, 12\right], 4 :\left[8, 9, 12\right], 5 :\left[8, 9, 10, 11, 12\right], 6 :\left[9, 10, 12\right], 7 :\left[12\right], 9 :\left[11\right]$ & $\left[0, \frac{1}{3}, 0, 0, \frac{2}{3}, \frac{1}{3}, \frac{1}{3}, \frac{2}{3}, \frac{1}{3}, \frac{1}{3}, \frac{2}{3}, \frac{1}{3}, 0\right]$ \\ \hline
			\texttt{L?\`{}FE\`{}wFdk]aNo} & $0 :\left[4, 6, 7, 10, 11\right], 1 :\left[5, 6, 7, 8, 11\right], 2 :\left[6, 8, 10, 11, 12\right], 3 :\left[7, 8, 9, 10, 11, 12\right], 4 :\left[8, 9, 12\right], 5 :\left[9, 10, 11, 12\right], 6 :\left[9, 10, 12\right], 7 :\left[12\right], 9 :\left[11\right]$ & $\left[0, \frac{1}{3}, 0, 0, \frac{2}{3}, \frac{1}{3}, \frac{1}{3}, \frac{1}{3}, \frac{1}{3}, \frac{1}{3}, \frac{2}{3}, \frac{1}{3}, \frac{1}{3}\right]$ \\ \hline
			\texttt{L?\`{}FF\`{}wFtk]aNo} & $0 :\left[4, 6, 7, 10, 11\right], 1 :\left[5, 6, 7, 8, 11\right], 2 :\left[6, 7, 8, 10, 11, 12\right], 3 :\left[7, 8, 9, 10, 11, 12\right], 4 :\left[8, 9, 12\right], 5 :\left[9, 10, 11, 12\right], 6 :\left[9, 10, 12\right], 7 :\left[9, 12\right], 9 :\left[11\right]$ & $\left[0, \frac{1}{3}, 0, 0, \frac{1}{3}, \frac{1}{3}, \frac{1}{3}, \frac{1}{3}, \frac{2}{3}, \frac{1}{3}, \frac{2}{3}, \frac{1}{3}, \frac{1}{3}\right]$ \\ \hline
			\texttt{L?\`{}NBQclQt\textasciitilde C\textasciitilde ?} & $0 :\left[4, 6, 8, 9, 11, 12\right], 1 :\left[5, 6, 7, 10, 11, 12\right], 2 :\left[6, 7, 8, 9, 11, 12\right], 3 :\left[9, 10, 11, 12\right], 4 :\left[5, 7, 10, 11, 12\right], 5 :\left[8, 9, 12\right], 6 :\left[10\right], 7 :\left[9\right], 8 :\left[10, 11\right], 9 :\left[10\right]$ & $\left[0, 0, 0, \frac{1}{3}, \frac{1}{3}, \frac{1}{3}, \frac{2}{3}, \frac{2}{3}, \frac{2}{3}, \frac{1}{3}, \frac{1}{3}, \frac{1}{3}, \frac{1}{3}\right]$ \\ \hline
			\texttt{L?pFEb\{Fd\{]aNo} & $\left\{0 : \left[4, 6, 7, 8, 10, 11\right], 1 : \left[4, 5, 6, 7, 8, 11\right], 2 : \left[6, 8, 10, 11, 12\right], 3 : \left[7, 8, 9, 10, 11, 12\right], 4 : \left[8, 9, 10, 12\right], 5 : \left[8, 9, 10, 11, 12\right], 6 : \left[9, 10, 12\right], 7 : \left[12\right], 9 : \left[11\right]\right\}$ & $\left[\frac{1}{3}, \frac{1}{3}, \frac{1}{3}, 0, \frac{1}{3}, 0, \frac{1}{3}, \frac{2}{3}, 0, \frac{2}{3}, \frac{1}{3}, \frac{1}{3}, \frac{1}{3}\right]$ \\ \hline
			\texttt{L?pFFb\{Ft\{]aNo}\quad($\overline{C_{13}^{3}}$) & $\left\{0 : \left[4, 6, 7, 8, 10, 11\right], 1 : \left[4, 5, 6, 7, 8, 11\right], 2 : \left[6, 7, 8, 10, 11, 12\right], 3 : \left[7, 8, 9, 10, 11, 12\right], 4 : \left[8, 9, 10, 12\right], 5 : \left[8, 9, 10, 11, 12\right], 6 : \left[9, 10, 12\right], 7 : \left[9, 12\right], 9 : \left[11\right]\right\}$ & $\left[\frac{1}{3}, \frac{1}{3}, \frac{1}{3}, \frac{1}{3}, \frac{1}{3}, \frac{1}{3}, \frac{1}{3}, \frac{1}{3}, \frac{1}{3}, \frac{1}{3}, \frac{1}{3}, \frac{1}{3}, \frac{1}{3}\right]$ \\ \hline
			\texttt{O?\`{}@?boNBoBsBy\`{}\}WZfAk} & $\left\{0 : \left[4, 8, 13, 14, 15\right], 1 : \left[5, 8, 10, 14, 15\right], 2 : \left[6, 8, 9, 10, 15\right], 3 : \left[7, 8, 9, 10, 11\right], 4 : \left[9, 10, 11, 12\right], 5 : \left[9, 11, 12, 13\right], 6 : \left[11, 12, 13, 14\right], 7 : \left[12, 13, 14, 15\right], 8 : \left[11, 12, 13\right], 9 : \left[13, 14, 15\right], 10 : \left[12, 13, 14\right], 11 : \left[14, 15\right], 12 : \left[15\right]\right\}$ & $\left[0, \frac{1}{3}, \frac{2}{3}, \frac{2}{3}, 1, \frac{2}{3}, \frac{1}{3}, \frac{1}{3}, 0, 0, 0, 0, 0, \frac{1}{3}, \frac{1}{3}, \frac{1}{3}\right]$ \\ \hline
			\texttt{OCrb\`{}roNBoBsBy\`{}\}WZfAk} & $\left\{0 : \left[3, 4, 5, 8, 13, 14, 15\right], 1 : \left[4, 5, 6, 8, 10, 14, 15\right], 2 : \left[5, 6, 7, 8, 9, 10, 15\right], 3 : \left[6, 7, 8, 9, 10, 11\right], 4 : \left[7, 9, 10, 11, 12\right], 5 : \left[9, 11, 12, 13\right], 6 : \left[11, 12, 13, 14\right], 7 : \left[12, 13, 14, 15\right], 8 : \left[11, 12, 13\right], 9 : \left[13, 14, 15\right], 10 : \left[12, 13, 14\right], 11 : \left[14, 15\right], 12 : \left[15\right]\right\}$ & $\left[\frac{2}{5}, \frac{1}{5}, \frac{2}{5}, \frac{1}{5}, \frac{2}{5}, \frac{1}{5}, \frac{2}{5}, \frac{1}{5}, \frac{1}{5}, 0, \frac{2}{5}, \frac{2}{5}, 0, \frac{2}{5}, 0, \frac{2}{5}\right]$ \\ \hline
			
		\end{tabular}
	\end{adjustbox}

\end{table}

\bibliographystyle{plain}      
\bibliography{bibliography}   

\end{document}